\documentclass[a4paper,UKenglish,cleveref, autoref, thm-restate]{lipics-v2021}

\usepackage{xcolor}
\usepackage{amssymb,amsmath,amsthm}
\usepackage{mathtools}
\DeclareMathAlphabet{\mathcalligra}{T1}{calligra}{m}{n}

\usepackage{tcolorbox}

\usepackage[linesnumbered, ruled, vlined]{algorithm2e}
\SetKwRepeat{Do}{do}{while}%

\SetCommentSty{mycommfont}

\usepackage[]{todonotes}

\usepackage{bbm}
\usepackage{wrapfig}
\usepackage{stmaryrd}

\usepackage{multirow}

\usepackage{colortbl}

\usepackage{tikz}
\usetikzlibrary{matrix}
\usetikzlibrary{arrows}
\usetikzlibrary{shapes}
\usetikzlibrary{automata, positioning,through,calc}

\tikzset{
->, 
>=stealth, 
node distance=2cm, 
every state/.style={thick, fill=gray!10,ellipse}, 
interm state/.style={thick, fill=gray!10, draw, rectangle, inner sep=6}, 
initial text=$ $, 
}

\newcommand{\unzip}{\mathrm{unzip}}

\newcommand{\alphabet}{\Sigma}
\newcommand{\domain}{\Sigma}
\newcommand{\Prop}{X}
\newcommand{\trace}{\tau}
\newcommand{\atrace}{\rho}

\newcommand{\val}{v}

\newcommand{\fininf}{\omega}
\newcommand{\pad}{\#}

\newcommand{\AllAsyncTn}{{(\alphabet^*)}^{\Prop}}

\newcommand{\lt}{\sigma}

\newcommand{\setTraces}{T}
\newcommand{\setATraces}{R}

\newcommand{\setsetTraces}{\textbf{T}}

\newcommand{\traceVar}{\pi}

\newcommand{\traceAssign}{\Pi}

\newcommand{\VarTrace}{\mathcal{V}}
\newcommand{\FOL}{\text{FO}}
\newcommand{\Rel}{\mathop{\precsim}}

\newcommand{\NodeFOL}{\FOL[\Rel]}

\newcommand{\Redux}[1]{\lfloor #1 \rfloor}

\newcommand{\sync}{\mathrm{zip}}
\newcommand{\async}{\mathrm{unzip}}
\newcommand{\Sync}[1]{\mathrm{Zip}(#1)}
\newcommand{\Async}[1]{\mathrm{Unzip}(#1)}
\newcommand{\As}{\mathop{:}}

\newcommand{\States}{Q}
\newcommand{\SFinal}{F}
\newcommand{\SInitial}{\hat{\States}}
\newcommand{\SLabel}{\gamma}
\newcommand{\Transition}{\delta}

\newcommand{\actpat}{p}

\newcommand{\Actions}{A}

\newcommand{\Lang}{\mathcal{L}}

\newcommand{\std}{q}
\newcommand{\initstd}{\hat{\std}}
\newcommand{\runA}{r}

\newcommand{\node}{\textit{node}}

\newcommand{\fhnaut}{\mathcal{H}}

\newcommand{\stutFreeA}{\mathcal{A}}

\newcommand{\alphNSFA}{\alphabet^{\Prop}}

\newcommand{\Pre}[2]{\textit{In}(#2)}
\newcommand{\Pos}[2]{\textit{Out}(#2)}

\newcommand{\deter}[1]{\text{det}(#1)}

\newcommand{\filter}[3]{#1^{#2}[#3]}

\newcommand{\compl}{\mathrm{complete}}

\newcommand{\uniNSFA}{\mathcal{U}}

\newcommand{\as}[2]{#1\mathop{:}#2}

\newcommand{\nextA}{\mathrm{next}}

\newcommand{\ProTraceVar}{\Prop_{\VarTrace}}

\newcommand{\Kpke}{K}
\newcommand{\KStates}{W}
\newcommand{\KState}{w}

\newcommand{\KTransition}{\Delta}
\newcommand{\KLabel}{V}

\newcommand{\KSFull}{\Kpke \mathop{=} (\KStates, \domain^\Prop, \KTransition,  \KLabel)}

\newcommand{\KPaths}{\mathrm{Paths}}

\newcommand{\Kpath}{\varrho}

\newcommand{\ActLab}{\mathbb{A}}

\newcommand{\KpkeOp}{\mathbb{W}}
\newcommand{\KSIn}{\KStates_{\mathrm{in}}}
\newcommand{\KSOut}{\KStates_{\mathrm{out}}}

\newcommand{\Slice}[1]{\mathrm{Slice}(#1)}

\newcommand{\Intersect}{\mathrm{Join}}
\newcommand{\toAct}{\mathop{\downarrow}_{\ActLab}}

\newcommand{\pInput}{\textbf{read}}
\newcommand{\pOutput}{\textbf{output}}

\newcommand{\pStatus}{\mathrm{status}}
\newcommand{\pVar}{\mathrm{var}}

\newcommand{\pDebugW}[1]{\textit{Deb}_{#1}}
\newcommand{\pClear}[1]{\textit{Clear}_{#1}}

\newcommand{\tAnd}{\text{ and }}
\newcommand{\tOr}{\text{ or }}

\newcommand{\tIf}{\text{ if }}
\newcommand{\tthen}{\text{ then }}

\newcommand{\tIff}{\text{ iff }}
\newcommand{\tSt}{\text{ s.t.\ }}

\newcommand{\Def}{\stackrel{\mathclap{\text{def}}}{=}}

\nolinenumbers

\title{Hypernode Automata} 

\author{Ezio Bartocci}{Technische Universit\"at Wien, Vienna, Austria \and \url{http://www.eziobartocci.com} }{ezio.bartocci@tuwien.ac.at}{https://orcid.org/0000-0002-8004-6601}{}

\author{Thomas A. Henzinger}{IST Austria, Klosterneuburg, Austria \and \url{http://pub.ist.ac.at/~tah/} }{tah@ist.ac.at}{https://orcid.org/0000-0002-2985-7724}{}

\author{Dejan Nickovic}{AIT Austrian Institute of Technology, Vienna, Austria }{dejan.nickovic@ait.ac.at}{https://orcid.org/0000-0001-5468-0396}{}%

\author{Ana Oliveira da Costa}{Technische Universit\"at Wien, Vienna, Austria }{ana.costa@tuwien.ac.at}{}{}

\authorrunning{E. Bartocci et al.} 

\Copyright{Ezio Bartocci and and Thomas A. Henzinger and Dejan Nickovic and Ana Oliveira da Costa} 



\category{} 

\relatedversion{} 

\EventEditors{}
\EventNoEds{0}
\EventLongTitle{}
\EventShortTitle{}
\EventAcronym{}
\EventYear{}
\EventDate{}
\EventLocation{}
\EventLogo{}
\SeriesVolume{}
\ArticleNo{}

\begin{document}

\maketitle

\begin{abstract}
We introduce \emph{hypernode automata} as a new specification formalism for hyperproperties of concurrent systems. 
They are finite automata with nodes labeled with \emph{hypernode logic} formulas and transitions labeled with actions. 
A hypernode logic formula specifies relations between sequences of variable values in different system executions. 
Unlike HyperLTL, hypernode logic takes an \emph{asynchronous} view on execution traces by constraining the values and the order of value changes of each variable without correlating the timing of the changes. 
Different execution traces are synchronized solely through the transitions of hypernode automata. 
Hypernode automata naturally combine asynchronicity at the node level with synchronicity at the transition level.
We show that the model-checking problem for hypernode automata is decidable over action-labeled Kripke structures, 
whose actions induce transitions of the specification automata.
For this reason, hypernode automaton is a suitable formalism for specifying and verifying asynchronous hyperproperties, such as declassifying observational determinism in multi-threaded programs.
\end{abstract}

\section{Introduction}
\label{sec:intro}

Formalisms like Linear temporal logic (LTL) or automata are commonly used to specify and verify trace properties of concurrent systems.
Security requirements such as information-flow policies require simultaneous reasoning about multiple execution traces; hence, they cannot be expressed as trace properties.
Hyperproperties address this limitation by specifying properties of trace sets~\cite{ClarksonS10}.  HyperLTL~\cite{ClarksonFKMRS14}, an extension of LTL with trace quantifiers, has emerged as a popular formalism for both the specification and verification of an important class of hyperproperties.
The temporal operators of HyperLTL and related hyperlogics progress in lockstep over all traces that are bound to a trace variable;
they specify \emph{synchronous} hyperproperties.
As a consequence, HyperLTL cannot specify, for instance, an information-flow policy that changes from one system mode to another if the mode transition can occur at different times in different system executions~\cite{BartocciFHNC22}.
This limitation has been observed repeatedly and independently in recent years by \cite{GutsfeldMO21,lics2021,cav2021}
all of whom have proposed asynchronous versions of hyperlogics to address the problem.

We take a different route and propose a specification language for hyperproperties, called \emph{hypernode automata}, which combines synchronicity and asynchronicity by combining automata and logic.
\emph{Hypernode automata} are finite automata with nodes labeled with formulas from a fully asynchronous, non-temporal hyperlogic, called \emph{hypernode logic}, and transitions labeled with actions used to synchronize different execution traces. 
While automata-based languages have been used before for specifying synchronous hyperproperties \cite{BonakdarpourS21}, 
hypernode automata are the first language that systematically separates trace synchronicity from trace asynchronicity in the specification of hyperproperties:
within hypernodes (i.e., states of the hypernode automata), different execution traces proceed at independent speeds,
only to ``wait for each other'' when they transition to the next hypernode.
This separation leads to natural specifications and plays to the strengths of both automata-based and logic-based formalisms.

Hypernodes' specification adopts a maximally asynchronous view over finite trace segments: each program variable can progress independently.
We introduce \emph{hypernode logic} to specify such asynchronous hyperproperties.
Hypernode logic includes quantification over finite 
traces and the binary relation $x(\pi)\Rel y(\pi')$, for trace variables $\pi$ and~$\pi'$, and system (or program) variables $x$ and $y$.
This relation asserts that the program variable $x$ undergoes the same
ordered value changes in the trace assigned to $\pi$ as the variable $y$ does in~$\pi'$,
but the changes may happen at different times (stuttering), 
and there may be additional changes of $y$ in $\pi'$ (prefixing). 
%
%
The \emph{stutter-reduced prefixing} relation $\Rel$ between finite traces is the only nonlogical operator
of hypernode logic, yielding a novel, elegant, and powerful method to specify asynchronous hyperproperties
over finite traces.

For each finite or infinite action sequence, a \emph{hypernode automaton} specifies a corresponding sequence of formulas from hypernode logic.
%
This paper's main contribution is a model-checking algorithm for hypernode automata.
Our algorithm checks if a given hypernode automaton accepts the set of (possibly) infinite traces defined by an action-labeled Kripke structure.
The action labels on transitions of the Kripke structure (the ``model'') induce equally labeled transitions of the hypernode automaton (the ``specification''). 
The subroutine that model-checks formulas of hypernode logic is technically novel: it introduces automata-theoretic constructions on a new concept called \emph{stutter-free automata}, which are then used in familiar logical contexts such as filtration and self-composition.
While hypernode logic has existential trace quantifiers and thus can specify nonsafety hyperproperties such as the independence of variables \cite{BartocciFHNC22}, we focus in this paper on safe hypernode automata to specify hyperproperties of infinite executions.
%
%
To our knowledge, this is the first decidability result for a simple but general formalism that can specify important asynchronous hyperproperties.

Perhaps the most famous asynchronous hyperproperty is 
\emph{observational determinism}~\cite{ZdancewicMyers2003}.
In Section~\ref{sec:example}, we motivate hypernode logic by providing a formal specification of observational determinism as defined by~\cite{ZdancewicMyers2003}. 
We further motivate hypernode automata by specifying an information-flow policy that requires dynamic changes in observational determinism between different modes of a concurrent program.
An example of such a dynamic change is the 
\emph{declassification} of information, which can be represented by a transition between hypernodes.
Section~\ref{sec:automata} defines 
hypernode logic 
and hypernode automata.
%
We demonstrate the useful expressive power of hypernode logic by 
providing formal specifications for several different versions of observational determinism that have been discussed in the literature~\cite{Huismanetal2006,Terauchi2008}.
In Section~\ref{sec:mc}, we first solve the model-checking problem for hypernode logic over Kripke structures using stutter-reduced automata, 
and then the more general model-checking problem for hypernode automata over action-labeled Kripke structures.  
In contrast to previously proposed specification formalisms for asynchronous hyperproperties \cite{GutsfeldMO21,lics2021,cav2021}, 
which are undecidable in general and decidable only for specific fragments, our model-checking algorithms are doubly exponential in the number of variables in the Kripke structure.
Finally, in Section~\ref{sec:related_work}, we argue that our formalism is expressively incomparable to these other formalisms.
For this reason, hypernode automaton is a promising specification formalism for asynchronous hyperproperties and thus significantly contributes to the automatic verification of security properties.

\subsubsection*{Summary of contributions}

We present a new specification formalism for hyperproperties,
with stateless (hypernode logic) and stateful (hypernode automata) specifications, which (1)~can be used to specify paradigmatic asynchronous information-flow properties such as 
observational determinism (stateless) and information declassification (stateful); (2)~is expressively orthogonal to temporal hyperlogics;
and (3)~has a decidable verification problem, solved by a new automata-theoretic approach using stutter-free automata.
While hypernode logic has existential trace quantifiers and thus can specify nonsafety hyperproperties such as the independence of variables \cite{BartocciFHNC22}, we focus in this paper on safe hypernode automata to specify hyperproperties of infinite executions.
While we expect the addition of accepting hypernodes, such as hypernode B\"uchi automata, to be straightforward, this is left to future work.

  \section{Motivating Example}
\label{sec:example}

\begin{figure}
\caption{Hypernode automaton $\fhnaut$ specifying the mutually exclusive declassification of secure information, where \(\varphi_{\text{od}}(L) \overset{\text{def}}{=} \forall \traceVar \forall\traceVar' \bigwedge_{l\in L} (l(\traceVar) \Rel l(\traceVar') \vee  l(\traceVar') \Rel l(\traceVar))\).}
\label{fig:spec:declass:vars}
 \centering
\vspace{-0.2cm}
\scalebox{0.9}{
	 \begin{tikzpicture}[thick]
	 \node[state, accepting] (s) {\(\varphi_{\text{od}}(\{y,z\})\)};
	 \node[state, right = 1.1cm of s, accepting] (sz) {\(\varphi_{\text{od}}(\{z\})\)};
	 \node[state, left = 1.1cm of s, accepting] (sy) {\(\varphi_{\text{od}}(\{y\})\)};
  \draw [->, thick] (0,1) to[] node{} (s);
	 \draw
	 	(s) edge[bend right=10, above] node{\(\pDebugW{z}\)}  (sy)
	 	(s) edge[bend left=10, above] node{\(\pDebugW{y}\)}  (sz)
	 	(sy) edge[bend right=10, below] node{\(\pClear{}\)}  (s)
	 	(sz) edge[bend left=10, below] node{\(\pClear{}\)}  (s)
	 	(s) edge[loop below, looseness=5] node{\(\pClear{}\)} (s)
	 	(sz) edge[loop right,align=center,looseness=5] node{\(\pDebugW{y}\)\\ \(\pDebugW{z}\)} (sz)
	 	(sy) edge[loop left,align=center,looseness=5] node{\(\pDebugW{y}\)\\ \(\pDebugW{z}\)} (sy)
	 	;
	 \end{tikzpicture}}
\end{figure}

The seminal work by Zdancewic and Myers~\cite{ZdancewicMyers2003} proposes the notion of observational determinism to specify information-flow policies of concurrent programs.
Observational determinism is a noninterference property which requires publicly visible values to not depend on secret information. 
Noninterference specification is particularly challenging for multi-threaded programs because (i)~the executions of a multi-threaded program depend on the scheduling policy, and (ii)~a change from one program state to another can happen at different times in each execution of the program.  
According to \cite{ZdancewicMyers2003}, a program is observationally deterministic if, when starting from any two low-equivalent states, then any two traces of each low variable are equivalent up to \emph{stuttering} and \emph{prefixing}. 
Huisman, Worah, and Sunesen~\cite{Huismanetal2006}  and Terauchi~\cite{Terauchi2008} suggested several variations of this definition of observational determinism, again based on stuttering and prefixing equivalences.
This definition takes an asynchronous view of execution traces, so HyperLTL is not adequate to specify it \cite{BartocciFHNC22}.

In this section, we use hypernode logic to specify 
\emph{Zdancewic-Myers observational determinism}, and use it to define a \emph{mutually exclusive declassification policy} with a hypernode automaton.
The declassification policy involves not only the asynchronous requirement of observational determinism, but also dynamic changes between different observational determinism requirements. 
In particular, the policy requires that during \emph{normal} operation,  
two publicly visible variables \(y\) and \(z\) must not leak secret information. 
The policy also admits two \emph{debugging} modes, in which  
either \(y\) or \(z\) can leak information, but never both. 
The ``mode operation'' is inspired by examples on declassification policies by the programming languages community (c.f. \cite{aslan2007gradualrelease,sabelfeld2009declassification}).

The hypernode automaton specification $\fhnaut$ of the declassification policy is shown in Figure~\ref{fig:spec:declass:vars}.
A hypernode automaton is interpreted over a set of action-labeled traces,
which are sequences of valuations for program variables and actions.
The transitions between automaton nodes are labeled with actions marking when the program changes its mode of operation, say, from normal mode to one of the two debugging modes.
Note that the specification automaton is complete and deterministic: 
any sequence of actions identifies a unique run of the automaton. 
In the example, transitions from the normal mode to the two debugging modes are labeled with $\pDebugW{y}$ and $\pDebugW{z}$ actions, respectively; transitions from either debugging mode to the normal mode are labeled with $\pClear{}$.

The automaton nodes are labelled with formulas of hypernode logic.
In our example, all three formulas specify Zdancewic-Myers observational determinism but for different program variables. 
Observational determinism requires that for any two program executions (specified by  \(\forall \traceVar \forall \traceVar'\)), their projections to each publicly visible program variable (\(l\) in a set \(L\) of public variables) are equivalent up to stuttering and prefixing (specified by  \(l(\traceVar) \Rel l(\traceVar') \vee l(\traceVar') \Rel l(\traceVar))\).
The visible program variables change from mode to mode, so the different specification nodes are labeled with different instances of the same formula.
For example, the hypernode with formula \(\varphi_{\text{od}}(\{y\})\) requires observational determinism only for \(y\).


Algorithm~\ref{alg:p_var} defines a reactive program \(\mathcal{P}_{\pVar}\) (where \(\pVar\) can be either \(y\) or \(z\)) which in every iteration reads the input variable $x$ and the action $\pStatus$. 
If, for example, the action is $\pDebugW{y}$, then variable \(y\) is used for debugging and the program copies the content of $x$ to $y$. 
%

\begin{figure}
\begin{minipage}[t]{0.4\textwidth}
\begin{algorithm}[H]
\caption{Program \(\mathcal{P}_{\pVar}\)}
\label{alg:p_var}
\small
\SetAlgoLined
\SetKwRepeat{Do}{do}{while}
 \Do{\(\mathrm{true}\)}{
  \(\pVar:=0\);\\
  \(\pInput(x)\);\\ 
  \(\pInput(\pStatus)\);\\
  \If{\em (\(\pStatus = \pDebugW{\pVar}\))}
  {\(\pVar := x\); 
  }
  \(\pOutput(\pVar)\)\;
 }
\end{algorithm}
\end{minipage}
\begin{minipage}[t]{0.51\textwidth}
\centering
\captionof{table}{Executions of $\mathcal{P}_{y}~||~\mathcal{P}_{z}$.}
\label{tab:traces_examples}
\scalebox{1}{
\begin{tabular}{c|r c c c c}
\hline 
\(\trace_1\)\ \  & $x$:\ \  &  0 & \cellcolor{gray!15}0 & \cellcolor{gray!50}0 &\\
& $y$:\ \ & 0 & \cellcolor{gray!15}0 & \cellcolor{gray!50}0 &\\
& $z$:\ \  & 0 & \cellcolor{gray!15}0 & \cellcolor{gray!50}0 &\\
& \(\pStatus\):\ \ &\(\varepsilon\) & \cellcolor{gray!15} \(\pDebugW{y}\) & \cellcolor{gray!50} \(\pDebugW{z}\) & \\ 
\hline 
\(\trace_2\) \ \ &  $x$:\ \ & 1 & 1 & \cellcolor{gray!15}1 & \cellcolor{gray!50} 1\\
&  $y$:\ \  &0 & 0 & \cellcolor{gray!15}1 & \cellcolor{gray!50} 1\\
&  $z$:\ \ & 0 & 0 & \cellcolor{gray!15}0 & \cellcolor{gray!50} 1\\
&\(\pStatus\):\ \  &\(\varepsilon\) &\(\varepsilon\) & \cellcolor{gray!15} \(\pDebugW{y}\) & \cellcolor{gray!50} \(\pDebugW{z}\) \\ 
\hline 
\end{tabular}}
\vspace{-0.5cm}
\end{minipage}
\end{figure}

The parallel composition $\mathcal{P}_{y}~||~\mathcal{P}_{z}$ does not satisfy the specification~$\fhnaut$.
Consider, for instance, the set  
$T = \{\tau_1, \tau_2\}$ of traces shown in Table~\ref{tab:traces_examples}.
We make two important observations on $\tau_1$ and $\tau_2$: 
(1) these traces have different lengths, and 
(2) they exhibit the \emph{same} sequence of actions ($\pDebugW{y}$ followed by $\pDebugW{z}$) 
happening at \emph{different} times (hence the traces are asynchronous). 
We note that the above 
sequence of actions partitions each 
trace into a sequence of three trace segments, called \emph{slices}, which 
we denote using the white, the light gray, and the dark gray color
in Table~\ref{tab:traces_examples}. 
We map each slice to a unique node in the hypernode automaton.
We first observe that the two traces 
(read from left to right) have different lengths are 
asynchronous in the sense that they have the same sequence 
of actions happening at different times.
We also observe that every sequence of actions partitions the trace set 
into a sequence of trace set segments, called \emph{slices}. 
Each slice is mapped to a unique node of the hypernode automaton. 
The white slice of $T$ is mapped to the initial node of $\fhnaut$. The 
light-gray slice of $T$ is mapped to the node accessible from the initial state 
with action $\pDebugW{y}$ (i.e., the debugging mode for \(y\)), 
which is labeled by the formula \(\varphi_{\text{od}}(\{z\})\). 
The dark-gray slice of $T$ is mapped to the same node, 
because the action $\pDebugW{z}$ triggers the self-loop transition.

A sequence of actions defines a path in the hypernode automaton.
Then, a set $T$ of traces with a sequence of 
actions \(p\) satisfies the hypernode automaton $\fhnaut$ iff the
slicing of $T$ induced by \(p\)
satisfies the hypernode formulas in the path defined by \(p\) in \(\fhnaut\).
This is not the case in our example because the dark-gray 
slice of $T$ violates its 
associated hypernode formula 
$\varphi_{\text{od}}(\{z\})$. 
More specifically, the program variable $z$ evaluates to $0$ in the dark-gray segment 
of the trace $\tau_1$, 
while it evaluates to $1$ in the dark-gray segment of~$\tau_2$.
The specification violation occurs because the critical section 
(lines $5-7$ in Algorithm~\ref{alg:p_var}) is unprotected. 
It is possible that the action $\pDebugW{z}$
happens (line $4$ in $\mathcal{P}_z$) after $\pDebugW{y}$.
Thereafter the input value $x$ is 
copied to $z$ (line $6$ in $\mathcal{P}_z$), and  both $y$ and $z$ are made observable (line $8$ in 
both $\mathcal{P}_y$ and $\mathcal{P}_z$).
Hence they both leak 
information about $x$, which violates the specification. 
 This problem can be remedied by introducing an appropriate 
 lock mechanism to protect the critical section, to ensure 
 that between two observations $x$ can never be copied to 
 both $y$ and~$z$. 
Our model-checking algorithm allows us to fully automate the
reasoning used to find the bug in this example.

  \section{Hypernode Automata}
  \label{sec:automata}

  In this section, we define \emph{hypernode logic}
  and \emph{hypernode automata.}
  We represent program executions as finite or infinite sequences of finite trace segments with synchronization actions.
Let \(\Prop\) be a finite set of \emph{program variables} over a finite domain \(\domain\),
$A$ be a finite set of \emph{actions} and $A_\varepsilon=A\cup\{\varepsilon\}$.

Hypernode logic is interpreted over finite trace segments.
A \emph{trace segment} \(\trace\) is a finite sequence of valuations in \(\domain^{\Prop}\), 
where each \emph{valuation} \(\val\As\Prop\rightarrow\domain\) maps program variables to domain values.
We denote the set of trace segments over \(\Prop\) and \(\domain\) by \((\domain^{\Prop})^*\). 
A \emph{segment property} \(\setTraces\) is a set of trace segments, 
that is, \(\setTraces \mathop{\subseteq} (\domain^{\Prop})^*\).
A formula of hypernode logic specifies a property of a set of trace segments, 
which is called a \emph{segment hyperproperty}.
Formally, a segment hyperproperty \(\setsetTraces\) is a set of segment properties, 
that is, \(\setsetTraces\mathop{\subseteq}2^{(\domain^{\Prop})^*}\).

Hypernode automata are interpreted over finite and infinite action-labeled traces.
An \emph{action-labeled trace} $\atrace$ is a finite or infinite sequence of pairs, 
each consisting of a valuation and either an action label from $A$, or the empty label $\varepsilon$;
that is,
\(\atrace\in (\domain^\Prop\! \times\! \Actions_\varepsilon)^*\) or
\(\atrace\in (\domain^\Prop\! \times\! \Actions_\varepsilon)^{\fininf}\).
We require, for technical simplicity, that for infinite action-labeled traces,
infinitely many labels are non-empty.
An \emph{action-labeled trace property} is a set of action-labeled traces.
A hypernode automaton accepts action-labeled trace properties,
and thus specifies an \emph{action-labeled trace hyperproperty},
namely, the set of all action-labeled trace properties it accepts.
  
  \subsection{Hypernode Logic}
  \label{sub:hypernode_logic}
  

Hypernode logic, \(\NodeFOL\), is a first-order formalism to specify relations between the changes of the values of program variables over a set of trace segments.
The formulas of hypernode logic are defined by the  grammar:
\(
\varphi\ ::=\ \exists \traceVar\, \varphi \, |\, \neg \varphi\, |\, \varphi \wedge \varphi \, |\, x(\traceVar) \Rel x(\traceVar),
\)
where the first-order variable \(\traceVar\) ranges over the set \(\VarTrace\) of trace variables and the unary function symbol \(x\) ranges over the set \(\Prop\) of program variables.
%
Hypernode logic refers to time only through the 
binary \emph{stutter-reduced prefixing} predicate $\Rel$.
The intended meaning of the atomic formula \({x(\traceVar) \Rel y(\traceVar')}\) is that $x$ undergoes the same ordered value changes in the trace segment assigned to $\traceVar$ as the variable $y$ does in~$\traceVar'$, followed by possibly additional value changes of $y$ in $\traceVar'$.
In other words, hypernode logic adopts a fully asynchronous comparison 
of different trace segments in which all variables are considered separately.

We therefore interpret hypernode formulas over \emph{unzipped} trace segments, which encode the evolution of each program variable independently. 
An \emph{unzipped trace segment} \(\trace\As \Prop \rightarrow \domain^{*}\) is a function from  the program variables to finite strings of values.
%
%
%
%
The formulas of hypernode logic are interpreted over assignments of trace variables to unzipped trace segments.
Given a set \({\setTraces \mathop{\subseteq} \AllAsyncTn}\) of unzipped trace segments, an assignment \(\traceAssign_{\setTraces}\As \VarTrace \rightarrow \setTraces\)  maps each trace variable to an unzipped trace segment in \(\setTraces\). 
We denote by \(\traceAssign_{\setTraces}[\traceVar \mathop{\mapsto} \trace]\)
the update of \(\traceAssign_{\setTraces}\), where \(\traceVar\) is 
assigned to \(\trace\).
%
%
The satisfaction relation for a formula \(\varphi\) of hypernode logic over an assignment 
\(\traceAssign_{\setTraces}\) is defined inductively as follows:
\vspace{-0.1cm}
\[
\begin{split}
&\traceAssign_{\setTraces} \models \exists \traceVar \varphi 
\tIff 
\text{ there exists }
\trace \in \setTraces:\ \traceAssign_{\setTraces}[\traceVar \mapsto \trace] \models \varphi;\\
&\traceAssign_{\setTraces} \models  \psi_1 \wedge \psi_2 
\tIff 
\traceAssign_{\setTraces} \models  \psi_1 \tAnd  \traceAssign_{\setTraces} \models   \psi_2; \ \ 
\traceAssign_{\setTraces} \models  \neg  \psi_1
\tIff 
\traceAssign_{\setTraces} \not\models  \psi_1 ;\\
&\traceAssign_{\setTraces} \models  x(\traceVar)\Rel  y(\traceVar') 
\tIff  
\traceAssign_{\setTraces}(\traceVar)(x)\mathop{\in}
\lt_0^+\ldots \lt_n^+ \tAnd
\traceAssign_{\setTraces}(\traceVar')(y)\mathop{\in}
\lt_0^+\ldots \lt_n^+ \alphabet^* \\
&\hspace{3.4cm} \text{with } \lt_i \neq \lt_{i+1}, \text{ for } 0\leq i < n.
\end{split}
\]
A set \(\setTraces\) of unzipped trace segments is a \emph{model} of the formula \(\varphi\), denoted by \(\setTraces \models \varphi\), iff there exists an assignment \(\traceAssign_{\setTraces}\) such that \(\traceAssign_{\setTraces}\models \varphi\). 
We adopt the usual abbreviations \(\forall \traceVar \varphi \Def \neg \exists \traceVar \neg \varphi\) and \(\varphi \vee \varphi' \Def {\neg (\neg \varphi \wedge \neg\varphi')}\). 
From now on, unless stated otherwise, program and trace variables are indexed by a natural number, i.e., \(\Prop = \{x_1, \ldots, x_m\}\) and \(\VarTrace=\{\traceVar_1, \ldots, \traceVar_n\}\).

\begin{example} We illustrate how to use hypernode logic by specifying four different variants of non-interference between program variables.
Zdancewic and Myers introduced in \cite{ZdancewicMyers2003} the first notion of observational determinism to capture non-interference for concurrent programs. 
They require that in every program execution, every publicly visible variable in a set \(L\) must be stutter-equivalent up to prefixing (i.e., one of the executions can have more value changes):
\[\forall \traceVar \forall \traceVar'\
 \bigwedge\limits_{l \in L} (l(\traceVar) \Rel l(\traceVar') \vee l(\traceVar') \Rel l(\traceVar)).\]

Later, Huisman, Worah, and Sunesen \cite{Huismanetal2006} strengthened the previous definition by requiring every publicly visible variable to be stutter-equivalent in all executions:
\[
\forall \traceVar \forall \traceVar' \ \bigwedge\limits_{l \in L} 
l(\traceVar) \Rel l(\traceVar').\]
The universal trace quantification in this formula forces \(\Rel\) to be symmetric, and thus both traces must have the same ordered changes for each visible public variable.

Our third variant of observational determinism is from Terauchi\cite{Terauchi2008}, requiring the set of all publicly visible variables to be stutter-equivalent up to prefixing:
\[\forall \traceVar \forall \traceVar' \
(L(\traceVar) \Rel L(\traceVar') \vee L(\traceVar') \Rel L(\traceVar)).\]
Note that, we can encode the values of a finite set of variables within a single variable called \(L\) because we interpreted hypernode formulas over arbitrary finite domains.

Finally, we specify independence (also known as generalized non-interference \cite{ClarksonS10}) as defined in \cite{BartocciFHNC22}.
Two program variables $x$ and $y$ are \emph{independent} iff whenever
a sequence of value changes for \(x\) is possible in some trace \(\traceVar\),  and a sequence of value changes for \(y\) is possible in some trace \(\traceVar'\), then also their combination $(x(\traceVar),y(\traceVar'))$ is possible in some trace. 
The formula for independence specifies that for every two traces (\(\traceVar\) and \(\traceVar'\)) there exists a third trace (\(\traceVar_{\exists}\)) that witnesses the combination $(x(\traceVar),y(\traceVar'))$ up to stuttering and prefixing:
\[
\forall \traceVar \forall \traceVar' \exists \traceVar_{\exists}\
 (x(\traceVar)\Rel x(\traceVar_{\exists}) \wedge  
 y(\traceVar') \Rel y(\traceVar_{\exists})). \null\hfill\vartriangleleft\] 
\end{example}

\subsubsection*{\bf Stutter-reduced trace segments}
\label{sub:pre}

We are interested in unzipped trace segments that are stutter-free, i.e., that do not repeat 
the same variable value in consecutive time points.
For a program variable $x$ and unzipped trace segment \(\trace\) with 
\(\trace(x)\mathop{\in}\lt_0^{+} \ldots \lt_n^{+}\) where \(\lt_i \neq \lt_{i+1}\) for \(i< n\),
the \emph{stutter-reduction} is \(\Redux{\trace(x)}=\lt_0\ldots \lt_n\). 
We extend this notion naturally to the stutter-reduction of \(\trace\)
by \(\Redux{\trace}(x)= \Redux{\trace(x)}\) for all program variables \(x\in\Prop\),
and to the
stutter reduction of a set \(\setTraces\) of unzipped trace segments by
\(\Redux{\setTraces} = \{ \Redux{\trace} \,|\, \trace \in \setTraces\}\).
We prove that formulas of hypernode logic cannot distinguish between a set of unzipped trace segments \(\setTraces\) and its stutter-reduction \(\Redux{\setTraces}\).

\begin{proposition}\label{thm:minimal_set_traces}
Let \(\setTraces \mathop{\subseteq} \AllAsyncTn\) be a set of unzipped trace segments and \(\varphi\) a 
formula of hypernode logic. Then,
\(\setTraces \models \varphi\) iff \(\Redux{\setTraces} \models \varphi\).
\end{proposition}

  \subsection{Hypernode Automata}
  \label{sub:hypernode_aut}

Hypernode automata are finite automata with states (called \emph{hypernodes}) labeled with formulas of hypernode logic and transitions labeled with actions.
A hypernode automaton reads a set \({\setATraces\! \subseteq\! (\domain^\Prop\! \times\! \Actions_\varepsilon)^{\fininf}}\) of action-labeled traces, and accepts some of these sets.

\begin{definition}
\label{def:fhn}
A deterministic, finite \emph{hypernode automaton} (HNA) over a set of actions \(\Actions\) and a set of program variables \(\Prop\) is a tuple \(\fhnaut \mathop{=} (\States, \initstd, \SLabel, \Transition)\),
where \(\States\) is a finite set of states  with \(\initstd\in \States\) being the initial state,
the state labeling function $\SLabel$ assigns a closed formula of hypernode logic over the program variables $X$ to each state in $Q$,
and the transition function \({\Transition: \States \times  \Actions  \rightarrow \States}\) is a total function assigning to each state and action a unique successor state. 
\end{definition}

We assume the totality and determinism of the transition function only for the simplicity of the technical presentation. 
A \emph{run} of the HNA \(\fhnaut\) is a finite or infinite sequence 
\(\runA \mathop{=} \std_0 a_0\, \std_1 a_{1}\, \std_2 a_2 \ldots\)
of alternating hypernodes and actions which starts in the initial hypernode \(\std_0 \mathop{=} \initstd\) and follows the transition function, i.e., 
\(\Transition(\std_i, a_i) \mathop{=} \std_{i+1}\) for all 
\(i\ge 0\). 
We refer to the corresponding sequence \(\actpat \mathop{=} a_0 a_1 a_2 \ldots\) of actions as the \emph{action sequence} of \(\runA\). 
Note that each action sequence defines a unique run of \(\fhnaut\).
%
We denote by \(\fhnaut[\actpat]\) the run that is defined by the action sequence \(\actpat\); an example is shown in Fig.~\ref{fig:ex:hypernode}.

The \emph{action sequence} of an action-labeled trace 
\(\atrace\mathop{=} (\val_0,a_0) (\val_1,a_1)(\val_2,a_2)\ldots\), 
where \(\val_i\mathop{\in} \domain^{\Prop}\) and \(a_i\mathop{\in}\Actions_\varepsilon\) for all $i\ge 0$, is the projection of the trace to its actions, with all empty labels \(\varepsilon\) removed;
that is, \(\atrace[\Actions]\mathop{=} a'_0 a'_1 \ldots\) with 
\(a_0 a_1\ldots \mathop{\in} a'_0 \varepsilon^* a'_1 \varepsilon^*\ldots\)
and $a'_i\in \Actions$ for all $i\ge 0$. 
Given a set \(\setATraces\) of action-labeled traces, 
the \emph{projection of $\setATraces$ with respect to a finite action sequence \(\actpat \mathop{\in} \Actions^*\)} is
\(
\setATraces[\actpat]\! =\! \{ \atrace \mathop{\in} \setATraces \,|\, \atrace[\Actions] \mathop{=} \actpat\, \actpat'  
\text{\ for some suffix\ } \actpat' \mathop{\in} \Actions^*\cup \Actions^\omega \}.
\)
%

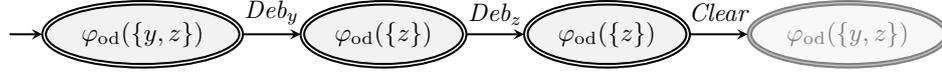
\begin{figure*}
\caption{Run defined by the action sequence \(\pDebugW{y}\pDebugW{z}\pClear{}\) in the hypernode automaton from Figure \ref{fig:spec:declass:vars}.}
\label{fig:ex:hypernode}
\centering
\scalebox{0.95}{
	 \begin{tikzpicture}[thick]
	 \node[state, initial, accepting] (s) {\(\varphi_{\text{od}}(\{y,z\})\)};
	 \node[state, right = 0.8cm of s, accepting] (sz1) {\(\varphi_{\text{od}}(\{z\})\)};
	 \node[state, right = 0.8cm of sz1, accepting] (sz2) {\(\varphi_{\text{od}}(\{z\})\)};
	 \node[state, right = 0.8cm of sz2, accepting,fill opacity=0.5, draw opacity=0.5] (sclear) {\(\varphi_{\text{od}}(\{y,z\})\)};
	 \draw
	 	(s) edge[align=center, above] node{\(\pDebugW{y}\)}  (sz1)
	 	(sz1) edge[align=center, above] node{\(\pDebugW{z}\)} (sz2)
	 	(sz2) edge[align=center, above] node{\(\pClear{}\)} (sclear)
	 	;
	 \end{tikzpicture}}
\end{figure*}

Each step in the run of a hypernode automaton defines a new \emph{slice} on a set of action-labeled traces.
Let $\actpat=a_0a_1\ldots a_n$ be a finite action sequence,
and $\atrace=(v_0,a'_0)(v_1,a'_1)\ldots$ be an action-labeled trace that has prefix $\actpat$, and let $\setATraces$ be a set of such traces.
We write $\atrace(\varnothing,a_0)$ for the initial trace segment of $\atrace$ which ends with the action label $a_0$.
Formally, $\atrace(\varnothing,a_0)=v_0\ldots v_k$ 
such that $a'_k=a_0$, and $a'_i=\varepsilon$ for all $0\le i<k$.
Furthermore, we write $\atrace(a_0a_1\ldots a_{i},a_{i+1})$ for the subsequent trace segments of $\atrace$ which end with the action label $a_{i+1}$ after having seen the action sequence $a_0a_1\ldots a_{i}$.
Inductively, if $\atrace(a_0a_1\ldots a_{i-1},a_i)=v_k\ldots v_l$,  $a'_m=a_{i+1}$ for \(m > l\), and $a'_j=\varepsilon$ for all $l\!<\!j\!<\!m$,
then $\atrace(a_0a_1\ldots a_i,a_{i+1})=v_{l+1}\ldots v_m$.
The slicing is extended to sets of action-labeled traces accordingly;
for example, 
\(\setATraces( \varnothing,a) \mathop{=} \{ \atrace( \varnothing, a)\,|\, \atrace \in \setATraces  \}\).

Since the formulas of hypernode logic are interpreted over unzipped trace segments, in a final step, we need to unzip each slice.
The \emph{{unzipping} of a trace segment \(\trace \mathop{=} \val_0 \ldots v_n\)} over the set of variables \(\Prop =\{x_1,\ldots, x_m\}\) is
\(\unzip(\trace)\mathop{=}\{\as{x_0}{\val_0(x_0)..\val_n(x_0)}, \ldots, \as{x_m}{\val_0(x_m)..\val_n(x_m)}\}.\)
We define the unzipping of trace segments sets naturally as \(\Async{\setTraces} \mathop{=} \{\async(\trace) \, |\, \trace\mathop{\in} \setTraces \}\).

\begin{definition}
\label{def:acceptance_hna}
Let \(\fhnaut \mathop{=} (\States, \initstd, \SLabel, \Transition)\) be an HNA, and \(\setATraces\) a set of action-labeled traces.
Let \(\actpat\) be a finite action sequence in \(\Actions^*\).
The set \(\setATraces\) is \emph{accepted} by \(\fhnaut\) with respect to the pattern \(p\), denoted \(\setATraces \models_{\actpat} \fhnaut\),
iff 
for the run \(\fhnaut[\actpat] \mathop{=} \std_0 a_0\, \std_1 a_1 \,\ldots\, \std_n a_n \), 
all slices of \(\setATraces\) induced by \(\actpat\) are models of the formulas that label the respective hypernodes;  
that is, 
\(\Async{\setATraces[\actpat](\varnothing,a_0)} \models \SLabel(\std_0)\), and 
\(\Async{\setATraces[\actpat](a_0 \ldots a_{i-1}, a_{i})} \models \SLabel(\std_{i})\) 
for all \(0< i \le n\). 
\end{definition}

A set \(\setATraces\) of action-labeled traces is \emph{accepted} by 
the HNA $\fhnaut$ iff for all finite action sequences \(\actpat \mathop{\in} \Actions^*\), 
if \(\setATraces[\actpat] \neq \emptyset\), then \(\setATraces \models_{\actpat} \fhnaut\). 
The \emph{language} accepted  by \(\fhnaut\) is the set of all sets of action-labeled traces 
that are accepted by \(\fhnaut\), denoted \(\Lang(\fhnaut)\).
Note that this definition assumes that all finite and infinite runs of HNA are feasible;
such automata are often called \emph{safety automata}.
Refinements are possible where finite runs must end in accepting states or, for example, 
infinite runs must visit accepting states infinitely often.

  \section{Model Checking}
  \label{sec:mc}
We solve the model-checking problem for hypernode automata over Kripke structures whose transitions are labeled with actions.
In other words, we give an algorithm that checks if an HNA accepts the set of all action-labeled traces generated by an action-labeled Kripke structure.
This result is surprising, as the verification problem is undecidable for many formalisms that allow the specification of asynchronous hyperproperties.
Also, our algorithm is substantially different from previously published model-checking algorithms for synchronous hyperproperties.

\subsection{Action-labeled Kripke Structures}
\label{sub:action_kripke}

A \emph{Kripke structure} is a tuple \(\KSFull\) consisting of a finite set \(\KStates\) of worlds, a set \(\Prop\) of variables over a finite domain \(\domain\),  a transition relation \({\KTransition\! \subseteq\! \KStates\! \times\! \KStates}\), and a value assignment
\({\KLabel\!:\!\KStates\times \Prop\! \rightarrow\! \domain}\) that assigns a value from the finite domain \(\domain\) to each variable in each world~\footnote{This definition of Kripke structures is equivalent to the usual definition over propositional variables because we can easily encode a set of variables over finite domains into a set of propositional variables.}.
Given a Kripke structure with a transition relation \(\KTransition\), and given a set $A$ of actions, an \emph{action labeling} for \(\Kpke\) over \(\Actions\) is a function \(\ActLab\As \KTransition \rightarrow 2^{\Actions_\varepsilon}\) that assigns a set of action labels (including possibly the empty label $\varepsilon$) to each transition.
A \emph{pointed Kripke structure} is a Kripke structure with one of its worlds being an initial world, denoted \((\Kpke, \KState_0)\) with \(\KState_0 \mathop{\in} \KStates\). 

A \emph{path} in the Kripke structure \(\Kpke\) with action labeling \(\ActLab\) is a finite or infinite sequence \(\KState_0 a_0\, \KState_1 a_1\, \KState_2 a_2 \ldots\) of alternating worlds and actions 
which respects both the transition relation, \((\KState_i, \KState_{i+1}) \mathop{\in} \KTransition\), and the action labeling, \(a_i \mathop{\in} \ActLab(\KState_i, \KState_{i+1})\), for all \(i \ge 0\). 
We write \(\KPaths(\Kpke,\ActLab)\) for the set of all such paths. 
The path \(\Kpath \mathop{=} \KState_0 a_0\, \KState_1 a_1 \ldots\) defines the
action-labeled trace \(\sync(\Kpath) \mathop{=} \KLabel(\KState_0) a_0 \,\KLabel(\KState_1) a_1\ldots\). 
We write \(\Sync{\Kpke,\ActLab}\) for the set of action-labeled traces defined by 
paths in \(\KPaths(\Kpke,\ActLab)\).
By \(\KPaths(\Kpke,\ActLab, \KState_0)\) we denote the set of all paths in \(\KPaths(\Kpke,\ActLab)\) 
that start at the world \(\KState_0\).
As before, \(\Sync{\Kpke,\ActLab,\KState_0}\) refers to the set of all action-labeled traces 
that are defined by paths in \(\KPaths(\Kpke,\ActLab,\KState_0)\).

We are now ready to formally define the central verification question solved in this paper, namely, the model-checking problem for specifications given as hypernode automata over models given as pointed Kripke structures with action labelling. 
The conversion of concurrent programs, such as those from Section~II, 
into a pointed Kripke structure with action labeling is straightforward; its formalization is omitted here for space reasons.

\begin{tcolorbox}[colback = gray!3, colbacktitle=gray!90, title={\textbf{Model-checking problem for hypernode automata}}]
Let \((\Kpke, \KState_0)\) be a pointed Kripke structure for set of variables \(\Prop\) over a finite domain \(\domain\), and 
let \(\ActLab\) be an action labeling for \(\Kpke\) over a set~\(\Actions\) of actions. 
Let \(\fhnaut\) be a hypernode automaton over the same set \(\Prop\) of variables, domain \(\domain\) and 
set \(\Actions\) of actions.
Is the set of action-labeled traces generated by $(\Kpke,\ActLab,w_0)$
accepted by $\fhnaut$;
that is, \(\Sync{\Kpke,\ActLab, \KState_0}\mathop{\in} \Lang(\fhnaut)\)?
\end{tcolorbox}

  \subsection{Model Checking Hypernodes}
  \label{sub:mc_hypernodes}

We begin by formulating and solving the model-checking problem for hypernode logic (rather than automata) over Kripke structures.
This algorithm constitutes the key subroutine for model-checking hypernode automata.
To interpret formulas of hypernode logic over a Kripke structure, 
we equip the Kripke structure with two set of worlds: the \emph{entry worlds}, where trace segments begin, and the \emph{exit worlds}, where trace segments end.
Formally, an \emph{open Kripke structure}
consists of a Kripke structure $\KSFull$,
and a pair $\KpkeOp=(\KSIn,\KSOut)$ consisting of 
a set $\KSIn\subseteq W$ of entry worlds, 
and a set $\KSOut\subseteq W$ of exit worlds. 

A \emph{path} of the open Kripke structure $(\Kpke,\KpkeOp)$ is path \(\KState_0 \ldots \KState_n\) in \(\Kpke\) that starts in a entry world, \(\KState_0 \mathop{\in} \KSIn\) and ends in an exit world, \(\KState_n \mathop{\in} \KSOut\).
As for Kripke structures,
we write \(\KPaths(K,\KpkeOp)\) for the set of paths of the open Kripke structure \((K,\KpkeOp)\).
The set of unzipped trace segments generated by the open Kripke structure \((\Kpke,\KpkeOp)\) 
is
\(\Async{\Kpke}(x)\mathop{=}\{ \KLabel(\KState_0,x) \ldots \KLabel(\KState_n,x) \,|\,\) \(\KState_0\!\ldots\!\KState_n \mathop{\in} \KPaths(\Kpke,\KpkeOp)\}\) for all variables \(x \in \Prop\).

\begin{tcolorbox}[colback = gray!3, colbacktitle=gray!90, title={\bf Model-checking problem for hypernode logic}]
Let \((\Kpke,\KpkeOp)\) be an open Kripke structure, and \(\varphi\) a formula of hypernode logic 
over the same set of variables \(\Prop\) and finite domain \(\domain\). 
Is the set of unzipped trace segments generated by $(\Kpke,\KpkeOp)$
a model for $\varphi$; 
that is, \(\Async{\Kpke,\KpkeOp}  \models \varphi\)?
\end{tcolorbox}

  \subsubsection*{\bf Stutter-free automata}

From Proposition~\ref{thm:minimal_set_traces}, it follows that it suffices to consider the stutter reduction of \(\Async{\Kpke,\KpkeOp}\) to solve the model-checking problem for hypernode logic.
We introduce \emph{stutter-free automata} as a formalism for specifying sets of stutter-free unzipped trace segments.
We use stutter-free automata boolean operators to define a filtration that, when applied to a hypernode formula $\varphi$ and a stutter-free automaton over the variables in \(\varphi\), returns an automaton with non-empty language iff the language of the input automaton is a model of $\varphi$.
Finally, we construct, from a given open Kripke structure $(\Kpke,\KpkeOp)$, a stutter-free automaton that accepts an unzipped trace segment if the segment is the stutter reduction of a trace segment generated by $(\Kpke, \KpkeOp)$.
All three parts put together define the model-checking algorithm for hypernode logic. 

Stutter-free automata are a restricted form of nondeterministic finite automata (NFA) 
that read unzipped trace segments and guarantees that, for each state, 
there are no repeated variable assignments on their incoming and outgoing transitions. 
%
We denote by \(\alphNSFA\) all assignments of variables in \(\Prop\) to values in \(\domain\) or the termination 
symbol \(\pad\). 
Formally, for \(\Prop \mathop{=} \{x_0, \ldots, x_m\}\), let 
\(\alphNSFA \mathop{=} \{\as{x_0}{\lt_0}, \ldots, \as{x_m}{\lt_m}\,|\, 
\forall 0 \leq i\leq m\ \lt_i \mathop{\in} \domain\cup\{\pad\} \} \setminus \{\as{x_0}{\pad}, \ldots, \as{x_m}{\pad}\}\). 

\begin{definition}
\label{def:stut_free_aut}
Let \(\Prop\) be a finite set of variables over \(\domain\). 
A nondeterministic \emph{stutter-free automaton} (NSFA) over the alphabet \(\alphNSFA\) is a tuple 
\(\stutFreeA \mathop{=} (\States, \SInitial, \SFinal,\Transition)\) with a finite set \(\States\) of states,
a set \(\SInitial\subseteq\States\) of initial states,  
a set \(\SFinal\subseteq \States\) of final states, 
and a transition relation \({\Transition: \States \times \alphNSFA \rightarrow 2^{\States}}\) 
that satisfies the following for all states \(\std \in \States\) and 
variables \(x \in \Prop\):
\begin{itemize}
\item 
\emph{stutter-freedom} requiring
\(\Pre{\Transition}{\std,x} \cap \Pos{\Transition}{\std,x} \subseteq \{\pad\}\), and
\item 
\emph{termination} requiring that if \(\pad\in \Pre{\Transition}{\std,x}\), then \(\Pos{\Transition}{\std,x} = \{\pad\}\),
\end{itemize}
where \(\Pre{\Transition}{\std,x}\) is the set of all \(x\)-valuations incoming to state \(\std\) and 
\(\Pos{\Transition}{\std,x}\) is the set of all \(x\)-valuations outgoing from state \(\std\); formally,
\(
\Pre{\Transition}{\std,x} \mathop{=} \{ v(x) \,|\, \std \in \Transition(\std',v ) \textit{\ for some\ } \std'\in\States\}
\)
and \(
\Pos{\Transition}{\std,x}\mathop{=} \{v(x) \, |\, \Transition(\std,v)\! \neq\! \emptyset\}
\). 
\end{definition}

A \emph{run} of the stutter-free automaton \(\stutFreeA\) 
is a finite sequence \(\std_0 v_0 \std_1 v_1 \ldots v_{n-1} \std_n\) of alternating states and 
variable assignments which starts with an initial state, \(q_0 \in \SInitial\), and satisfies the 
transition function, 
\(\std_{i+1} \in \Transition(\std_i, v_{i})\) for all \( i < n\). 
The run is \emph{accepting} if it ends in a final state, \(q_n\in F\). 
An unzipped trace segment \(\trace\) over a set of variables \(\Prop\)  with domain \(\domain\) is \emph{accepted} by the 
stutter-free automaton \(\stutFreeA\) iff there exists an accepting run \(\std_0 v_0 \ldots v_{n-1} \std_n\) 
such that \(\trace(x) \mathop{=} v_0(x) \ldots v_{n-1}(x)\), for all \(x\in \Prop\).
The \emph{language} of \(\stutFreeA\), denoted \(\Lang(\stutFreeA)\), is the set of all accepted unzipped trace segments accepted by \(\stutFreeA\).
We sometimes refer to the language of a stutter-free automaton without the termination symbol: 
\(\Lang(\stutFreeA)|_{\pad} \mathop{=} \{\trace|_{\pad} : X\rightarrow \domain^* \ |\ \trace \in \Lang(\stutFreeA) \}\), 
where \(\trace|_{\pad}\) removes all occurrences of \(\pad\) in a trace segment \(\trace\). 
Note that since $\stutFreeA$ is stutter-free,
$\Lang(\stutFreeA)|_{\pad}=\Redux{\Lang(\stutFreeA)|_{\pad}}$,
where $\Redux{\cdot}$ is the stutter reduction of unzipped trace segments.

In Figure \ref{fig:ex:loop2}, we show an example of a stutter-free automaton that accepts 
all unzipped trace segments for the boolean variables \(\{x,y\}\), 
where all \(x\)-trace segments are of odd size, 
while all \(y\)-trace segments are of even size, 
and the first value for both \(x\) and \(y\) is \(0\). 
Then, the automaton \(\stutFreeA\) from Figure \ref{fig:ex:loop2} defines the language
\(
\Lang(\stutFreeA)\! =\! \{ \as{x}{\trace_x}, \as{y}{\trace_y} \, |\,  t_x \in (01)^*0 \tAnd t_y \in (01)^*01 \}.
\)
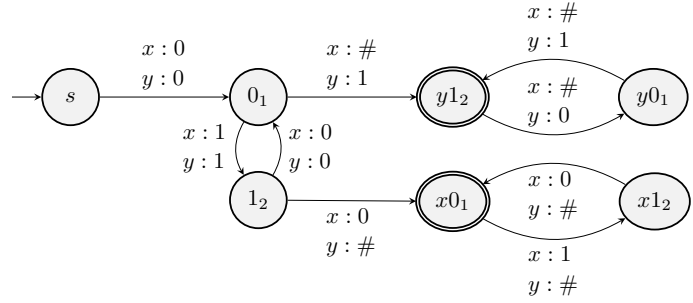
\begin{figure}[t]
 \caption{Stutter-free automaton \(\stutFreeA\), where \(x\)-trace segments have odd size, 
 and \(y\)-trace segments have even size, and the first values of \(x\) and \(y\) are \(0\).}
 \label{fig:ex:loop2}
  \centering
  \scalebox{0.85}{
	 \begin{tikzpicture}[]
	\node[state, initial] (s) {\(s\)};
	\node[state, right = 2cm of s] (s0) {\(0_1\)};
	\node[state, below = 0.7 cm of s0] (s12) {\(1_2\)};
 
	\node[state, accepting, right = 2cm of s0] (sy12) {\(y1_2\)};
	\node[state, right = 2cm of sy12] (sy01) {\(y0_1\)};
	
	\node[state, accepting, below = 0.7cm of sy12] (sx01) {\(x0_1\)};
	\node[state, right = 2cm of sx01] (sx12) {\(x1_2\)};

	 \draw
	 	(s) edge[above,text width=0.7cm] node{\(x:0\)\\\(y:0\)}  (s0)
	 	
	 	(s0) edge[bend right,text width=0.7cm, left] node{\(x:1\)\\\(y:1\)}  (s12)
	 	(s12) edge[bend right,text width=0.7cm, right] node{\(x:0\)\\\(y:0\)}  (s0)
	 	
	 	(s0) edge[text width=0.8cm, above] node{\(x:\pad\)\\\(y:1\)}  (sy12)
	 	
	 	(sy01) edge[bend right,text width=0.8cm, above] node{\(x:\pad\)\\\(y:1\)}  (sy12)
	 	(sy12) edge[bend right,text width=0.8cm, above] node{\(x:\pad\)\\\(y:0\)}  (sy01)

	 	(s12) edge[text width=0.8cm, below] node{\(x:0\)\\\(y:\pad\)}  (sx01)
	 	
	 	(sx01) edge[bend right,text width=0.8cm, below] node{\(x:1\)\\\(y:\pad\)}  (sx12)
	 	(sx12) edge[bend right,text width=0.8cm, below] node{\(x:0\)\\\(y:\pad\)}  (sx01)
	 	;
	 \end{tikzpicture}}
\end{figure}

The union, intersection, and determinization for NSFA are as usual for NFA.
Formally, for stutter-free automata over the alphabet \(\alphNSFA\), \(\stutFreeA_1 = (\States_1, \SInitial_1, \SFinal_1, \Transition_1)\) and
\(\stutFreeA_2 = (\States_2, \SInitial_2,\) \(\SFinal_2,\Transition_2)\), their \emph{union} is \(\stutFreeA_1 \cup \stutFreeA_2 = (\States_1 \dot{\cup} \States_2, \SInitial_1  \dot{\cup} \SInitial_2, \SFinal_1 \dot{\cup} \SFinal_2, \Transition_{\cup})\)
over the same alphabet \(\alphNSFA\)
where \(\Transition_{\cup}(\std) = \Transition_i(\std)\) when \(\std \mathop{\in} \States_i\) with \(i \mathop{\in} \{1,2\}\);
and their \emph{intersection} is 
\(\stutFreeA_1 \cap \stutFreeA_2 = (\States_1 \times \States_2, \SInitial_{\cap}, \SFinal_{\cap}, \Transition_{\cap})\) over the same alphabet \(\alphNSFA\)  where \(\SInitial_{\cap} = \{(\std_1, \std_2)\ |\ \std_1 \in \SInitial_1 \tAnd \std_2 \in \SInitial_2 \}\),
\(\SFinal_{\cap} = \{(\std_1, \std_2)\ |\ \std_1 \in \SFinal_1 \tAnd \std_2 \in \SFinal_2 \}\) and 
\(\Transition_{\cap}((\std_1, \std_2), l) = (\std_1', \std_2')\) iff
\(\Transition_1(\std_1, l) = \std_1'\) and \(\Transition_2(\std_2, l) = \std_2'\).
While the \emph{determinization} of \(\stutFreeA_1\) is  defined as \(\deter{\stutFreeA _1}= (2^{\States_1}, \SInitial_1, \SFinal_d, \Transition_d)\) where 
\(F_d\!=\!\{S\! \in\! 2^{\States_1} \, |\, S \cap F_1\! \neq\! \emptyset \}\) and  
\({\Transition_d(S, v)\!=\!\bigcup\limits_{\std \in S} \Transition_1(\std, v)}\) with \({v\in \alphNSFA}\).
We prove below that stuter-free automata is closed under the operators defined above.

\begin{proposition}
\label{thm:suttfree:bi_boolean_closure}
Let \(\stutFreeA_1\) and \(\stutFreeA_2\) 
be two stutter-free automata over the propositional variables \(X\).
Then, both \(\stutFreeA_1 \cup \stutFreeA_2\) and 
\(\stutFreeA_1 \cap \stutFreeA_2\) are stutter-free automata over $X$ with 
\(\Lang(\stutFreeA_1 \cup \stutFreeA_2) \mathop{=} \Lang(\stutFreeA_1) \cup \Lang(\stutFreeA_2)\) and \(\Lang(\stutFreeA_1 \cap \stutFreeA_2) \mathop{=} \Lang(\stutFreeA_1) \cap \Lang(\stutFreeA_2)\).
Moreover, 
for a nondeterministic stutter-free automaton \(\stutFreeA\),
the determinization \(\deter{\stutFreeA}\) is a deterministic 
stutter-free automaton over \(X\) with the same language,
that is, 
\(\Lang(\deter{\stutFreeA})\mathop{=} \Lang(\stutFreeA)\).
\end{proposition}

\begin{proof}
The first part follows from a direct translation from stutter-free automaton to NFA.
The union and disjunction does not affect the stutter-free related restrictions.

Consider an arbitrary NSFA, \(\stutFreeA_1\).
We start by proving that \(\deter{\stutFreeA_1}\) satisfies the stutter-free condition, i.e.\ 
\(\Pre{\Transition_d}{S,x} \cap \Pos{\Transition_d}{S,x} = \emptyset\), for all its states S and variables \(x \in \Prop\).
By definition:
\[
\begin{split}
\Pre{\Transition_d}{S,x} &= 
\{ v(x) \ |\ S \mathop{\in} \Transition_d(S',v ) \textit{\ for some\ } S'\in\States_d\}
\Leftrightarrow\\
\Pre{\Transition_d}{S,x} &= 
\{ v(x)  \ | \ S = \bigcup\limits_{\std_1 \mathop{\in} S_1} \Transition(\std_1, v ) \textit{\ for some\ } S'\in\States_d\}  \Leftrightarrow\\
\Pre{\Transition_d}{S,x} &= 
\{ v(x)\ |\ \forall \std \mathop{\in} S\ \exists \std_1 \mathop{\in} S' \tSt \std \in \Transition(\std_1, v) \textit{\ for some\ } S'\mathop{\in}\States_d\}.
\end{split}
\]
Thus, \((\star)\) for all \(\std \in S\), 
\(\Pre{\Transition}{\std,x} = \Pre{\Transition}{S,x}\). 
From \(\stutFreeA_1\) being a NSFA we know that, for all \(\std \in S\), 
\(\Pre{\Transition}{\std,x} \cap \Pos{\Transition}{\std,x} = \emptyset\).
Assume towards a contradiction that there exists a value that is in both the incoming and outgoing transitions of \(S\) for a variable \(x\), i.e.\ 
\(l \in \Pre{\Transition_d}{S,x} \cap \Pos{\Transition_d}{S,x}\).
Then, by definition of \(\Pos{\Transition_d}{S,x}\),
there exists a state in \(S\), \(\std \in S\), s.t.\
\(\Transition(\std, x:l) \neq \emptyset\). This contradicts our conclusion \((\star)\).
We prove now that \(\deter{\stutFreeA_1}\) satisfies the \(\pad\)-ending requirement. By \((\star)\), it follows that if
\(\pad \in \Pre{S}{x}\) then \(\pad \in \Pre{\std}{x}\) for all \(\std \in S\). So, by \(\stutFreeA\) being a NSFA it follows that for all \(\std \in S\) we have \(\Pos{\std}{x} = \{\pad\}\). Therefore \(\Pos{S}{x} = \{\pad\}\).

Finally, \(\Lang(\deter{\stutFreeA_1})= \Lang(\stutFreeA)\).
follows directly from the same result for NFA.
\end{proof}

The complementation of a stutter-free automaton follows the same approach as for NFA: we first determinize the automaton, then complete it, and lastly swap the final and nonfinal states. 
The only operation that requires special attention for NSFA is \emph{completion}, as we need to be careful to statisfy the condition of stutter-freedom.

A stutter-free automaton \({\stutFreeA \mathop{=} (\States, \SInitial, \SFinal, \Transition)}\) over \(\alphNSFA\) is \emph{complete} iff
\(\Pre{\Transition}{\std} \cup \Pos{\Transition}{\std}\) is a maximal subset of \(\alphNSFA\) according to the conditions in Definition \ref{def:stut_free_aut},
where \(\Pre{\Transition}{\std} = \{\val(x)\,|\ {x \mathop{\in} \Prop} \tAnd v(x) \mathop{\in} \Pre{\Transition}{\std, x}\}\) and \(\Pos{\Transition}{\std} = \{\val(x)\,|\, x \mathop{\in} \Prop \tAnd  v(x) \mathop{\in} \Pos{\Transition}{\std, x}\}\).
The \emph{universal} stutter-free automaton $\uniNSFA_{\alphNSFA}$ over $\alphNSFA$, defined next, is a deterministic and complete automaton with language $\Lang(\uniNSFA_{\alphNSFA})|_{\pad}=\Redux{(\domain^*)^{\Prop}}$, i.e., it
contains all stutter-free unzipped traces over $\alphNSFA$.
We use the universal stutter-free automaton as a ``sink'' area when completing other automata.

\begin{definition}
\label{def:stutter_free_universal}
Let $\Prop=\{x_0,\ldots,x_m\}$ be a set of variables over the finite domain \(\domain\).
The \emph{universal stutter-free automaton} over \(\alphNSFA\) is 
\(\uniNSFA_{\alphNSFA} \mathop{=} (\States_\uniNSFA, \States_\uniNSFA, \States_\uniNSFA, \Transition_{\uniNSFA})\), where
\(\States_\uniNSFA \mathop{=} \alphNSFA\) and
\begin{equation*}
\begin{split}
\Transition_{\uniNSFA}(\{&\as{x_i}{\lt_i}\}_{i\in [0,m]},\\  
\{&\as{x_i}{\lt'_i}\}_{i \in [0,m]})\mathop{=}
\end{split}
\begin{cases}
\{\as{x_i}{\lt'_i}\}_{i\in[0,m]} & \text{if } \forall 0 \leq i \leq m,  \tIf \lt_i = \pad \tthen \lt'_i \mathop{=} \pad \text{ else } \lt_i \neq \lt'_i;\\
\emptyset & \text{otherwise.}
\end{cases}
\end{equation*}
\end{definition}

Figure \ref{fig:ex:complete_final} shows the universal stutter-free automaton $\uniNSFA_X$ for the set of boolean variables \(X=\{x,y\}\).
 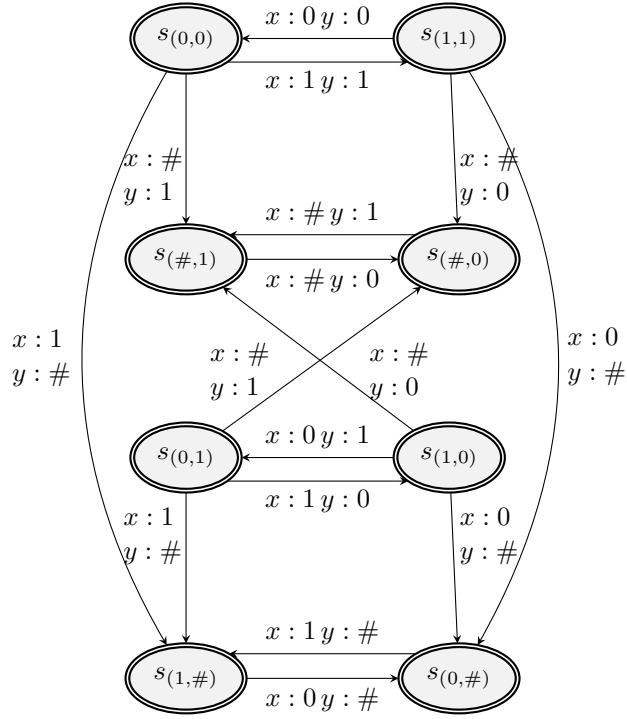
\begin{figure}[t]
   \centering
   \scalebox{1}{
 	 \begin{tikzpicture}[]
 	 \node[state, accepting] (s00) {\(s_{(0,0)}\)};
 	 \node[state, right = 2cm of s00,, accepting] (s11) {\(s_{(1,1)}\)};
 	 \node[state, below = 2cm of s00, accepting] (se1) {\(s_{(\pad,1)}\)};
 	 \node[state, right = 2cm of se1, accepting] (se0) {\(s_{(\pad,0)}\)};
 	 
 	 \node[state, below = 1.7cm of se1, accepting] (s01) {\(s_{(0,1)}\)};
 	 \node[state, right = 2cm of s01, accepting] (s10) {\(s_{(1,0)}\)};
 	 \node[state, below = 2cm of s01, accepting] (s1e) {\(s_{(1,\pad)}\)};
 	 \node[state, right = 2cm of s1e, accepting] (s0e) {\(s_{(0,\pad)}\)};

 	 \draw
 	 	(s00) edge [below,bend right, looseness = 0] node{\(x:1\,y:1\)} (s11)
 	 	(s11) edge [above] node{\(x:0\,y:0\)} (s00)
 	 	(se0) edge [above,bend right, looseness = 0] node{\(x:\pad\,y:1\)} (se1)
 	 	(se1) edge [below] node{\(x:\pad\,y:0\)} (se0)
 	 	
 	 	(s00) edge [text width =  0.8cm, left] node[label={[xshift=0.1cm, yshift=-1cm]{\(x:\pad\)\\\(y:1\)}}] {} (se1)
 	 	(s11) edge [text width =  0.8cm, right] node[label={[xshift=-0.05cm,yshift=-1cm]{\(x:\pad\)\\\(y:0\)}}] {} (se0)

 	 	(s01) edge [below,bend right, looseness = 0] node{\(x:1\,y:0\)} (s10)
 	 	(s10) edge [above] node{\(x:0\,y:1\)} (s01)
 	 	(s0e) edge [above,bend right, looseness = 0] node{\(x:1\,y:\pad\)} (s1e)
 	 	(s1e) edge [below] node{\(x:0\,y:\pad\)} (s0e)
 	 	
 	 	(s01) edge [text width =  0.8cm, left] node[label={[xshift=0.1cm,yshift=-0.2cm]{\(x:1\)\\\(y:\pad\)}}] {} (s1e)
 	 	(s10) edge [text width =  0.8cm, right] node[label={[xshift=-0.05cm,yshift=-0.2cm]{\(x:0\)\\\(y:\pad\)}}] {} (s0e)

 	 	(s00) edge [bend right,text width = 0.8cm, looseness = 1, left] node{\(x:1\)\\\(y:\pad\)} (s1e)
 	
 	 	(s11) edge [bend left,text width =  0.8cm, looseness = 1,right] node{\(x:0\)\\\(y:\pad\)} (s0e)	
 	 	
 	 	(s01) edge [right,text width = 0.8cm] node[label={[xshift=0.5cm,yshift=-0.8cm]{\(x:\pad\)\\\(y:0\)}}]{} (se0)
 	 	(s10) edge [left,text width = 0.8cm] node[label={[xshift=-0.5cm,yshift=-0.8cm]{\(x:\pad\)\\ \(y:1\)}}]{} (se1)
 	 	;
 	 \end{tikzpicture}}
  \caption{The universal stutter-free automaton $\uniNSFA_{\{x,y\}}$ over the boolean variables $\{x,y\}$.
  It accepts all stutter-free unzipped trace segments over $\{x,y\}$.
  All states are both initial and final.}
  \label{fig:ex:complete_final}
 \end{figure}
We use the states and transitions of the universal automaton to complete other stutter-free automata.
Formally, given a stutter-free automaton \(\stutFreeA \mathop{=} (\States , \SInitial, \SFinal, \Transition)\) over the variables $\Prop$ with domain \(\domain\), the \emph{completion} of \(\stutFreeA\) is the stutter-free automaton \(\compl(\stutFreeA) \mathop{=} (\States \mathop{\cup} \States_\uniNSFA,\SInitial, \SFinal, \Transition')\)
over the same variables, where $\States_\uniNSFA$ are the states of the universal stutter-free automaton $\uniNSFA_{\alphNSFA}$, and for all states \(q\mathop{\in} \States \cup  \States_\uniNSFA\) and all valuations \(v \in \alphNSFA\):
\begin{equation*}
\Transition'(\std, v) \mathop{=} 
\begin{cases}
\Transition_{\uniNSFA}(q,v) & \tIf q \in \States_\uniNSFA,\\
\Transition(q,v)\cup \{v'\ |\  \text{for all } x \mathop{\in} \Prop \tIf v(x)=\pad \tthen v'(x) = \pad &  \tIf q \in \States.\\
 \qquad\qquad\qquad\qquad   \text{otherwise } v'(x)\in (\domain \setminus \Pre{\Transition}{q,x} \cup \Pos{\Transition}{q,x})\}
&
\end{cases}
\end{equation*}
where \emph{In} and \emph{Out} are defined with respect to the transition relation $\Transition$ of $\stutFreeA$, 
and $\Transition_\uniNSFA$ is the transition relation of the universal stutter-free automaton $\uniNSFA_{\alphNSFA}$.

\begin{proposition}
\label{thm:complete_stutt}
Let \(\stutFreeA\) be a stutter-free automaton.
Then, \(\compl(\stutFreeA)\) is a complete stutter-free automaton with the same language as 
\(\stutFreeA\), i.e., \(\Lang(\compl(\stutFreeA))\mathop{=} \Lang(\stutFreeA)\).
\end{proposition}
\begin{proof}
Consider arbitrary stutter-free automaton \(\stutFreeA\).
Both \(\stutFreeA\) and the universal automaton are stutter-free automata, satisfying the stutter-free and termination requirements.
To prove that \(\compl(\stutFreeA)\) is a stutter-free automaton is only missing to prove that extension of the \(\stutFreeA\) transition relation
(pointing to states in the universal automaton) preserves both requirements.
By definition, if a variable trace is terminated, it will remain terminated (\(\tIf v(x)=\pad \tthen v'(x) = \pad\)), so termination is satisfied; and 
if it is not terminated, we only add the transitions labeled with values not seen in the incoming and outgoing edges of a state
(for all \(x\in \Prop\), \(v'(x) \notin \Pre{\Transition}{q} \cup \Pos{\Transition}{q}\)), hence it satisfies the stutter-free requirement.
We now prove that \(\compl(\stutFreeA)\) is complete. 
By definition of universal automaton, it follows that universal automata are complete. 
For all non-complete transitions in \(\stutFreeA\) (i.e., variables that are not terminated yet), we add the missing transition pointing to the matching universal automaton state; hence \(\compl(\stutFreeA)\) is complete.

Finally, we prove that both automata define the same language.
As we kept all \(\stutFreeA\) transitions, it follows that 
\({\Lang(\compl(\stutFreeA)) \supseteq \Lang(\stutFreeA)}\).
We now prove that \(\Lang(\compl(\stutFreeA)) \subseteq \Lang(\stutFreeA)\).
First, we note that no transition connects states from the universal automaton to states from \(\stutFreeA\). 
Then, when a run reaches a state from the universal automaton, all the following steps are within the universal automaton. 
Additionally, the final states do not include states from the universal automaton. 
Thus, accepting runs include only transitions in  \(\stutFreeA\).
\end{proof}

The complement of a deterministic and complete stutter-free automaton \(\stutFreeA \mathop{=} (\States, \SInitial, \SFinal, \Transition)\) over \(\alphNSFA\) is \(\overline{\stutFreeA} \mathop{=} (\States , \SInitial, \States\setminus\SFinal, \Transition)\), 
with the final and nonfinal states interchanged.

\begin{proposition}
\label{prop:stutter_free_compl}
Let  \(\stutFreeA\) be a deterministic and complete stutter-free automaton over \(\alphNSFA\).
Then, \(\overline{\stutFreeA}\) is a stutter-free automaton and 
\(\Lang(\overline{\stutFreeA}) \mathop{=} \Redux{\AllAsyncTn}\setminus\Lang(\stutFreeA)\).
\end{proposition}
\begin{proof}
Let \(\stutFreeA\) be an arbitrary stutter-free automaton.
It follows directly from  \(\stutFreeA\) being a stutter-free automaton that
\(\overline{\stutFreeA}\) is also stutter-free automaton (complementing an automaton does not change its transition function).
For the same reason, \(\overline{\stutFreeA}\) is a deterministic and complete stutter-free automaton, as well.
Then, for all unzipped traces \(\trace \mathop{\in} \Redux{\AllAsyncTn}\) there exists a run \(r\) in \(\overline{\stutFreeA}\) for \(\trace\) (from completeness); and it this the only run for \(\trace\) (from determinism).
As \(\overline{\stutFreeA}\) and \(\stutFreeA\) share the same transition function, then \(r\) is also only run in \(\stutFreeA\) for \(\trace\).
Finally, as the final states are flipped, it follows that
\(r\) is an accepting run for \(\trace\) in \(\overline{\stutFreeA}\) iff 
\(r\) is not an accepting run for \(\trace\) in \(\overline{\stutFreeA}\).
Hence
\(\Lang(\overline{\stutFreeA}) = \Redux{\AllAsyncTn}\setminus\Lang(\stutFreeA)\).
\end{proof}

\subsubsection*{\bf From formulas of hypernode logic to stutter-free automata}
Having prepared the ground by defining stutter-free automata, which are closed under union, intersection, and complement, we now turn to the model-checking problem for hypernode logic. 
Given a stutter-free automaton and a hypernode formula,  we define an inductive filtration (i.e., in each step we get produce a sub-automaton) over the hypernode formula structure to apply to the automaton we want to model-check.
The input automaton is a model of the input formula if the language of the automaton returned by the filtration is non-empty.
The inductive filtration for boolean operators translates naturally to automata operators.
For atomic hypernode formulas  (i.e., the predicate $\Rel$), we define a stutter-free automaton that captures the meaning of $\Rel$.
We depict an overview of the
model-checking algorithm in Fig. \ref{fig:mc_hypernode_logic}.

\begin{figure}
	\centering
	\scalebox{0.70}{ 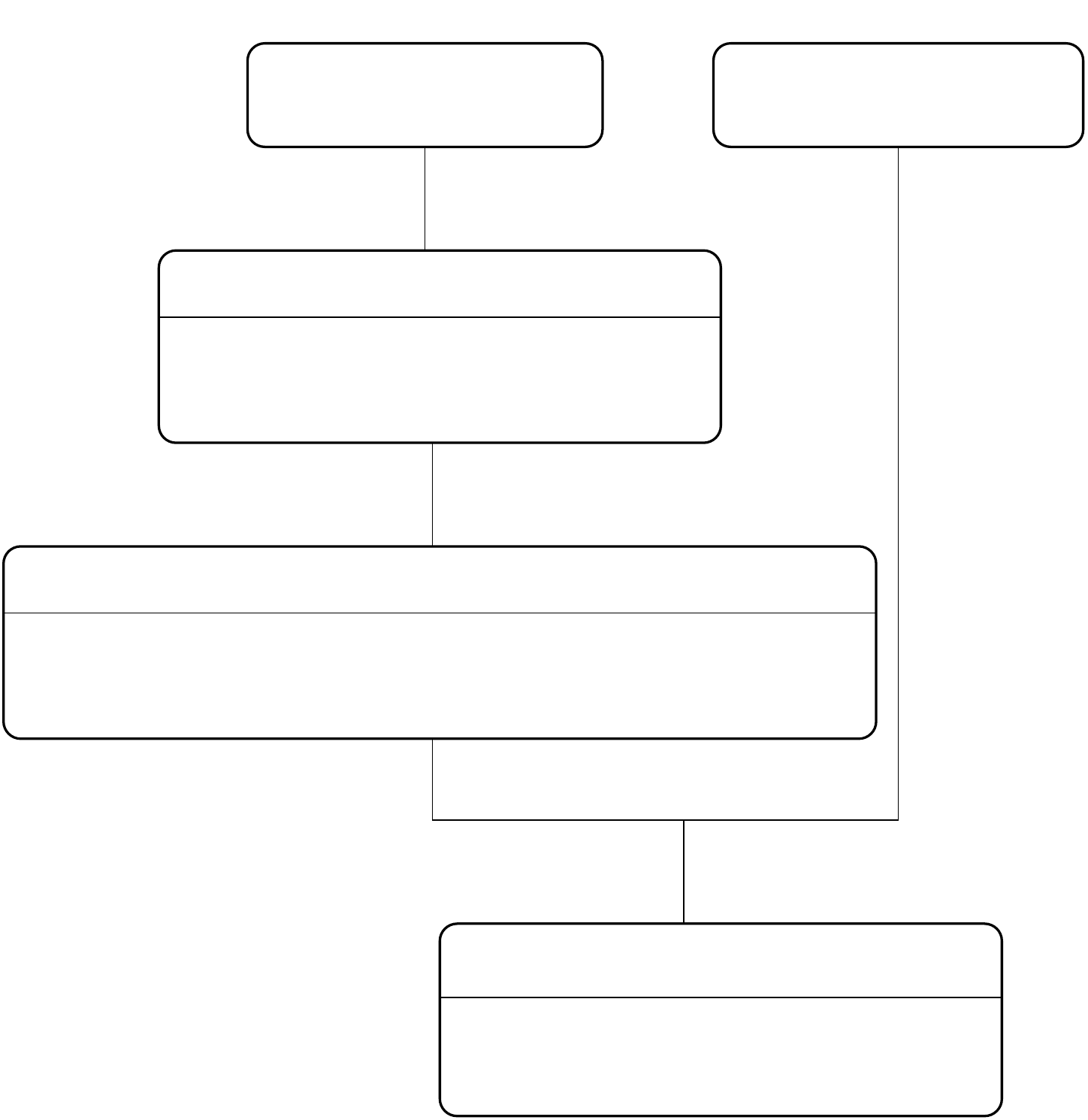 }
	\caption{Model-checking algorithm for hypernode formulas with relevant results.}
	\label{fig:mc_hypernode_logic}
\end{figure}

\begin{definition}
Let \(\uniNSFA_X \mathop{=} (\States_\uniNSFA, \States_\uniNSFA, \States_\uniNSFA, \Transition_{\uniNSFA})\)  
be the universal stutter-free automaton over \(\alphNSFA\),
and let $x,y\in \Prop$.
The \emph{stutter-free automaton for the atomic formula \(x \Rel y\) of hypernode logic} 
is the stutter-free automaton \(\stutFreeA_{x \Rel y} \mathop{=} (\States, \States, \States , \Transition)\) over the same variables and domain,
where \(\States\mathop{=} \{v\mathop{\in} \States_\uniNSFA \,|\, v(x) \mathop{=}  v(y) \tOr 
v(x)=\pad \}\), and 
\(\Transition(q,v) \mathop{=} \Transition_{\uniNSFA}(q,v)\)
for all $q\in Q$ and $v\in \alphNSFA$.
\end{definition}

To cope with trace variables in hypernode formulas,  we extend the set of variables $\Prop$ with a reference to trace variables in \(\VarTrace\), by \(\ProTraceVar \mathop{=} \{ x_{\traceVar} \,|\, x \in \Prop \tAnd \traceVar \in \VarTrace\}\). 
From an unzipped trace segment \(\trace\) over the set of variables \(\ProTraceVar\) we derive the trace assignment \(\traceAssign_{\trace}(\traceVar, x) \mathop{=} \trace(x_{\traceVar})\) for all \(x\in\Prop\) and \(\traceVar \in \VarTrace\).
We prove now that all words accepted by the stutter-free automaton for \(x(\traceVar) \Rel y(\traceVar')\) define assignments that satisfy that hypernode atomic formula.

\begin{lemma}
\label{lemma:pred_formula}
An unzipped trace segment \(\trace\) over \((\domain^{\ProTraceVar})^*\) is accepted by \(\stutFreeA_{x_\traceVar \Rel y_{\traceVar'}}\) over the same variables and domain,  \(\trace \in \Lang(\stutFreeA_{x_\traceVar \Rel y_{\traceVar'}})|_{\pad}\),
iff \(\traceAssign_{\trace} \models x(\traceVar) \Rel y(\traceVar')\).
\end{lemma}

\begin{proof}
Consider arbitrary unzipped trace segment \(\trace\) over \((\domain^{\ProTraceVar})^*\).
We start with the \(\Rightarrow\)-direction and assume that \(\trace \in \Lang(\stutFreeA_{x_\traceVar \Rel y_{\traceVar'}})|_{\pad}\).
Then, there exists an accepting run
\(\std_0 \lt_0 \std_1 \lt_1 \ldots \lt_{n} \std_{n+1}\) where, for all \(x_{\traceVar}\in \ProTraceVar\), if \(0\le i< |\trace(x_{\traceVar})|\), then
\(\lt_i(x_{\traceVar}) = \trace(x_{\traceVar},i)\), otherwise \(\lt_i(x_{\traceVar}) = \pad\) .
By definition of universal stutter-free automaton and \(\stutFreeA_{x_\traceVar \Rel y_{\traceVar'}}\), it follows that for all \(0\leq j \leq n\), \(\lt_j(x_{\traceVar})\mathop{=}\lt_j(y_{\traceVar'})\) or \(\lt_j(x_{\traceVar})\mathop{=}\pad\).
Moreover, by the termination requirement of stutter-free automaton,
 if there exists \(j\) s.t.\ \(\lt_j(x)\mathop{=}\pad\), then the value of \(x\) will not change anymore (i.e, for all \(k \ge j\),  \(\lt_k(x)\mathop{=}\pad\)).
We can then prove by induction of the trace \(\trace\) size that  \(\traceAssign_{\trace} \models x(\traceVar) \Rel y(\traceVar')\).
For the \(\Leftarrow\)-direction, we assume that \(\trace\) is not accepted by \(\stutFreeA_{x_\traceVar \Rel y_{\traceVar'}}\). 
As all states in \(\stutFreeA_{x_\traceVar \Rel y_{\traceVar'}}\) are final, then there exists a step \(0 \leq i \leq n\) where \(\lt_i(x_{\traceVar})\neq \lt_i(y_{\traceVar'})\) and \(\lt_i(x_{\traceVar}) \neq \pad\).
Given that \(\trace\) is stutter-free, then \(\traceAssign_{\trace} \not\models x(\traceVar) \Rel y(\traceVar')\).
\end{proof}

The inductive filtration defined next is the main element of the model-checking algorithm for formulas of hypernode logic.

\begin{definition}
\label{def:filter}
Let  \(\stutFreeA\) be a stutter-free automaton, and \(\varphi\) a formula of hypernode logic.
We define the positive and negative filtration of \(\stutFreeA\) by \(\varphi\), denoted \(\filter{\varphi}{+}{\stutFreeA}\) and  \(\filter{\varphi}{-}{\stutFreeA}\),respectively, inductively over the structure of $\varphi$ as follows:
\begin{align*}
\filter{(x(\traceVar) \Rel y(\traceVar'))}{+}{\stutFreeA} &\mathop{=} \stutFreeA \cap \stutFreeA_{x_\traceVar \Rel y_{\traceVar'}}
& \filter{(x(\traceVar) \Rel y(\traceVar'))}{-}{\stutFreeA} &\mathop{=} \stutFreeA \cap \overline{\stutFreeA_{x_\traceVar \Rel y_{\traceVar'}}}\\
\filter{(\varphi_1 \wedge \varphi_2)}{+}{\stutFreeA} &\mathop{=} \filter{\varphi_1}{+}{\stutFreeA} \cap  \filter{\varphi_2}{+}{\stutFreeA}
& (\varphi_1 \wedge \varphi_2)^-[\stutFreeA] &\mathop{=} \varphi_1^-[\stutFreeA] \cup  \varphi_2^-[\stutFreeA]\\
\filter{(\neg \varphi)}{+}{\stutFreeA} &\mathop{=} \filter{\varphi}{-}{\stutFreeA}
& \filter{(\neg \varphi)}{-}{\stutFreeA} &\mathop{=} \filter{\varphi}{+}{\stutFreeA}\\
\filter{(\exists \traceVar \varphi)}{+}{\stutFreeA} &\mathop{=} \filter{\varphi}{+}{\stutFreeA}
& \filter{(\exists \traceVar \varphi)}{-}{\stutFreeA} &\mathop{=} \stutFreeA \setminus \filter{\varphi}{+}{\stutFreeA}.
\end{align*}
\end{definition}

We reduce the problem of model checking a stutter-free automaton \(\stutFreeA\) over a formula \(\varphi\) with \(n\) trace variables to filtering the \(n\)-self-composition of \(\stutFreeA\) by~\(\varphi\).
The stutter-free automaton \(\stutFreeA^n\) is the result of composing \(n\) copies of \(\stutFreeA\) under a standard synchronous product construction, where for each copy \(\stutFreeA_i\), with \(i\le n\), all program variables \(x\in \Prop\) are renamed to \(x_{\traceVar_i}\).
Note that, the assignment derived by an unzipped trace segment \(\trace\) accepted by \(\stutFreeA^n\) defines a trace assignment from \(\{\traceVar_1, \ldots, \traceVar_n\}\) to traces accepted by \(\stutFreeA\).

\begin{theorem}
\label{thm:mc_hypernode_logic}
Let  \(\stutFreeA\) be a stutter-free automaton, and \(\varphi\) 
a formula of hypernode logic with \(n\) trace variables. Then,
\(\Lang(\filter{\varphi}{+}{\stutFreeA^n}) \neq \emptyset\) iff \(\Lang(\stutFreeA) \models \varphi\), and 
\(\Lang(\filter{\varphi}{-}{\stutFreeA^n}) \neq \emptyset\) iff \(\Lang(\stutFreeA) \not\models \varphi\).
\end{theorem}

\begin{proof}
Follows from the lemma we prove below for open formulas:\\

\emph{Let \(\trace\) be an unzipped trace segment over \((\domain^{ \ProTraceVar})^*\).
Then,
\(\trace \in \Lang(\filter{\varphi}{+}{\stutFreeA^n})|_{\pad}\) iff 
\(\traceAssign_{\trace} \models \varphi\); and 
\(\trace \in \Lang(\filter{\varphi}{-}{\stutFreeA^n})|_{\pad}\) iff
\(\traceAssign_{\trace} \not\models \varphi\).}\\

We prove the lemma by induction on the structure of the hypernode formula.
To prove the base case, \(x_\traceVar \Rel y_{\traceVar'}\) we use Lemma \ref{lemma:pred_formula}.
In particular, \(\trace \in \Lang(\filter{x_\traceVar \Rel y_{\traceVar'}}{+}{\stutFreeA^n})|_{\pad}\) iff \(\trace \in \Lang(\stutFreeA^n)|_{\pad}\) and \(\trace\in \Lang(\stutFreeA_{x_\traceVar \Rel y_{\traceVar'}})|_{\pad} \Leftrightarrow\).
As \(\trace \in \Lang(\stutFreeA^n)|_{\pad}\), then \(\traceAssign_{\trace}\) defines a trace assignment from \(\{\traceVar_1, \ldots \traceVar_n\}\) to traces accepted by \(\stutFreeA\) 
and, by Lemma \ref{lemma:pred_formula}, \(\traceAssign_{\trace} \models x_\traceVar \Rel y_{\traceVar'}\).
We prove analogously for \(\trace \in \Lang(\filter{x_\traceVar \Rel y_{\traceVar'}}{-}{\stutFreeA^n})|_{\pad}\).\\

We proceed to the inductive steps and assume as IH that the property holds for arbitrary hypernode formulas \(\varphi\) and \(\varphi'\).
The inductive case \(\varphi\wedge \varphi'\) follows from IH and Proposition \ref{thm:suttfree:bi_boolean_closure}.
While the inductive case \(\neg \varphi\) follows from IH and Proposition \ref{prop:stutter_free_compl}.\\

The case for the existential quantifier -- \(\exists \traceVar \varphi\) -- is more challenging.
We first prove the positive filtration and start with the \(\Rightarrow\)-direction.
Consider arbitrary  \(\trace \in \Lang(\filter{(\exists \traceVar \varphi)}{+}{\stutFreeA^n})|_{\pad}\). 
By definition of filtration, this is equivalent to \(\trace \in \Lang(\filter{\varphi}{+}{\stutFreeA^n})|_{\pad}\).
By IH, \(\traceAssign_{\trace} \models \varphi\), and so 
\(\traceAssign_{\trace} \models \exists \traceVar \varphi\).
We now prove the \(\Leftarrow\)-direction.
Assume a trace assignment exists s.t.\ \(\traceAssign \models \exists \traceVar \varphi\). 
Then, by definition of the satisfaction relation for hypernode formulas, there exists an unzipped trace \(\trace'\) over \(\domain^{\Prop}\) accepted by \(\stutFreeA\) (i.e., \(\trace' \in \Lang(\stutFreeA)|_{\pad}\)) s.t.\
\(\traceAssign[\traceVar \mapsto \trace'] \models \varphi\).
Note that the function to derive a trace assignment (mapping trace variables \(\VarTrace\) to unzipped traces over \(\domain^{\Prop}\)) 
from an unzipped trace segment \(\trace\) over \((\domain^{\ProTraceVar})^*\) is invertible.
Then, by IH, the trace \(\trace_{\traceVar}\) derived from the assignment extension \(\traceAssign_{\traceVar} = \traceAssign[\traceVar \mapsto \trace']\) is in \(\Lang(\filter{\varphi}{+}{\stutFreeA^n})|_{\pad}\).
And, by definition of filtration, \(\trace_{\traceVar} \in \filter{\exists \traceVar \varphi}{+}{\stutFreeA^n}\).

We now prove the negative filtering.
We start with the \(\Rightarrow\)-direction, which we prove by contra-position.
Assume that for an arbitrary trace assignment \(\traceAssign\) we have \(\traceAssign \models \exists \traceVar \varphi\), then there exists an extension \(\traceAssign_{\traceVar} = \traceAssign[\traceVar \mapsto \trace']\) with \(\trace'\in \Lang(\stutFreeA)|_{\pad}\) s.t.\ \(\traceAssign_{\traceVar} \models  \varphi\).
By IH, the trace \(\trace_{\traceVar}\) derived by \(\traceAssign_{\traceVar}\) is accepted by \(\filter{\varphi}{+}{\stutFreeA^n}\). 
Thus, \(\trace_{\traceVar}\) is not accepted by \(\stutFreeA^n\setminus \filter{\varphi}{+}{\stutFreeA^n}\), and, by definition, 
\(\trace_{\traceVar} \notin \Lang(\filter{(\exists \traceVar \varphi)}{-}{\stutFreeA^n})\).
For the \(\Leftarrow\)-direction,
consider arbitrary trace assignment s.t.\ \(\traceAssign \models \neg \exists \traceVar \varphi\). 
Then,  \(\traceAssign \models  \forall \traceVar \neg \varphi\). 
Equivalently, for all extensions of \(\traceAssign\), i.e.\ \(\traceAssign_{\traceVar} = \traceAssign[\traceVar \mapsto \trace']\) for all \(\trace' \in \Lang(\stutFreeA)|_{\pad}\), we have \(\traceAssign_{\traceVar} \not\models \varphi\). 
Consider an arbitrary of such extensions \(\traceAssign_{\traceVar}\), then, by IH, for the trace \(\trace_{\traceVar}\)  derived by it 
we have  \(\trace_{\traceVar} \in \Lang(\stutFreeA^n)|_{\pad}\) and \(\trace_{\traceVar} \notin \Lang(\filter{\varphi}{+}{\stutFreeA^n})\). 
Equivalently, \(\trace_{\traceVar} \in \Lang(\stutFreeA^n \setminus \filter{\varphi}{+}{\stutFreeA^n})\) and, so
\(\trace_{\traceVar} \in \Lang(\filter{(\exists \traceVar \varphi)}{-}{\stutFreeA^n})\).
\end{proof}

\subsubsection*{\bf From open Kripke structures to stutter-free automata}

We show now how to translate an open Kripke structure to a stutter-free automaton accepting the same unzipped trace segments as the set generated by the Kripke structure.
As unzipped traces represent the progression of each variable independently then, while building the stutter-free automaton from an open Kripke structure, we need to keep an independent progress pointer for each variable (to skip stuttering states).

We define first a function that returns the next stutter-free transition for each variable \(x\) from state \(\KState\) of a given Kripke structure.
Formally, for a Kripke structure \(\Kpke\mathop{=} (\KStates, \Prop, \KTransition,  \KLabel)\) over the set of variables \(\Prop\) with domain \(\domain\), the set with all worlds reachable from \(\KState\mathop{\in} \KStates\) with a transition to a world where \(x\mathop{\in}\Prop\) is not labeled with \(\lt\mathop{\in} \domain\cup\{\pad\}\) is:
\[N(\KState, \lt, x)\mathop{=}\{\KState'\ |\ \Async{\Kpke, \KState, \KState'}[x] \mathop{\in} \lt^+, (\KState', \KState'') \mathop{\in} \KTransition \tAnd \KLabel(\KState'',x) \neq \lt\},\]
where \(\Async{\Kpke, \KState, \KState'}\) is the set of unzipped traces  defined by paths in \(\Kpke\) from \(\KState\) to \(\KState'\).
To account for variables not having more reachable worlds with a different value, in the derived stutter-free automaton, we add a set of terminated states: \(\KStates^{\pad} \mathop{=} \{\KState^{\pad}\,|\, \KState \mathop{\in} \KStates\}\).
Finally, the translated stutter-free automaton transition relation uses the function \(\nextA/3\) defined below, adding transitions to terminated states as needed.
Formally,  we define \(\nextA(\KState\mathop{\in} \KStates\cup \KStates^{\pad}, \lt\mathop{\in} \domain\cup\{\pad\}, x\mathop{\in}\Prop)\) as follows:
\begin{equation*}
   	\nextA(\KState, \lt, x) = 
   	\begin{cases}
   		\{\KState^{\pad}\} 	& \tIf \KLabel(\KState,x)\mathop{=}\lt \tAnd N(\KState, \lt, x) \mathop{=} \emptyset,\\
   		\{\KState\} 	& \tIf \lt\mathop{=}\pad \tAnd \KState \mathop{\in} \KStates^{\pad} ,\\
   		N(\KState, \lt, x) 	& \text{ otherwise.}
   	\end{cases}
\end{equation*}

\begin{definition}
Let \((\Kpke,\KpkeOp) \mathop{=} ((\KStates, \domain^\Prop, \KTransition,  \KLabel),(\KSIn, \KSOut))\) be an open Kripke structure with \(\Prop\mathop{=}\{x_1, \ldots, x_m\}\).
The \emph{stutter-free automaton induced by \((K,\KpkeOp)\)} is
\(\stutFreeA_{K,\KpkeOp} \mathop{=} (\States, \SInitial,\) \(\SFinal, \Transition)\), 
where:
\begin{itemize}
\item 
\(\States \mathop{=} (\KStates \cup \KStates^{\pad})^{m}\);
\item 
\(\SInitial \mathop{=}\{(\KState_1, \ldots, \KState_m)\,|\, \KState_i \mathop{\in} \KSIn \text{ for all } 1 \leq i \leq m \}\);
\item 
\(\SFinal \mathop{=} \{ (\KState_1, \ldots, \KState_{m}) \ |\ \KState_i \in \KSOut \cup \KSOut^{\pad} \text{ for all } 1 \leq i \leq m \}\); 
\item 
\(
\Transition((\KState_1, \ldots, \KState_{m}), v)
= \{(\KState_1', \ldots, \KState_{m}')\,|\,\KState_i'\in \nextA(\KState_i, v[i], x_i)\)  for all \(i\ \}\).
\end{itemize}
\end{definition}

We prove that the language of the stutter-free automaton $\stutFreeA_{\Kpke,\KpkeOp}$
\((\Kpke,\KpkeOp)\) is the same as the stutter reduction of the set of unzipped trace segments defined by \((\Kpke,\KpkeOp)\).
\begin{lemma}
\label{thm:Kripke_lang_Redux_Stutter_free_lang}
For all open Kripke structures \((\Kpke,\KpkeOp)\), we have 
\(\Lang(\stutFreeA_{\Kpke,\KpkeOp})|_{\pad} =  \Redux{\Async{\Kpke,\KpkeOp}}\).
\end{lemma}
\begin{proof}
\(\Lang(\stutFreeA_{\Kpke,\KpkeOp})|_{\pad} \subseteq  \Redux{\Async{\Kpke,\KpkeOp}}:\)
Let \(\trace\) be an unzipped trace segment over \((\domain^{ \Prop})^*\) accepted by \(\stutFreeA_{\Kpke,\KpkeOp}\).
Then, there exists an accepting \(\stutFreeA_{\Kpke,\KpkeOp}\) run \(\std_0 \lt_0 \std_1 \lt_1 \ldots \lt_{n} \std_{n+1}\) where \(\lt_i = \trace[i]\), for \(i\ge n\).
We prove that the property holds by induction on the size of run independently for each variable  \(x_i\) with \(i\leq m\).
By definition of \(\nextA\), we prove that for the current step \(j\leq n\) -- \(\std_{j}[i]\ \lt_{j}[i]\ \std_{j+1}[i]\) -- either (i) \(N(\std_{j}[i], \lt_{j}[i], x_i) \mathop{\neq} \emptyset\) and then there exists a path \(\Kpath\) from \(\std_{j}[i]\) to \(\std_{j+1}[i]\) where the value of \(x_i\) remains \(\lt_{j}[i]\); (ii) \(N(\std_{j}[i], \lt_{j}[i], x_i) \mathop{=} \emptyset\) and the next state is terminated (\(\std_{j+1}[i] \mathop{=} \std_{j}^{\pad}[i] \)); or 
(iii) \(\std_{j}[i]\) is a terminated stated and so \(\lt_{j}[i] \mathop{=} \pad\) and \(\std_{j+1}[i] \mathop{=} \std_{j+1}^{\pad}[i]\).
Using this we build a path in the Kripke structure by extending the partial path given by the stutter-free automata with \(\Kpath\) when (i)  holds and not doing anything when (ii) and (iii) holds.
Note that we made sure to include the terminated exited paths in the final states of the stutter-free automaton, guaranteeing the the accepting condition matches the definition of a complete path over an open Kripke structure.
When we stutter reduce the new path we get \(\trace[x_i]\).

\(\Lang(\stutFreeA_{\Kpke,\KpkeOp}) \supseteq  \Redux{\Async{\Kpke,\KpkeOp}}:\) Let \(\trace\) be an unzipped stutter-free trace segment in \linebreak
\(\Redux{\Async{\Kpke,\KpkeOp}}\).
Then, there exists a path in \((\Kpke,\KpkeOp)\) that generates that trace segment.
From that path, for each variable, we build a sequence of wolds that removes all stutter transitions.
We prove then, by induction, that that sequence defines an accepting run of \(\stutFreeA_{\Kpke,\KpkeOp}\).
\end{proof}

\subsubsection*{\bf Model checking hypernode logic over Kripke structures}
With the reduction of open Kripke structures to stutter-free automata we can now use the filtration from Definition \ref{def:filter} to solve the model-checking problem for hypernode logic.

\begin{theorem}
\label{thm:mc_hypernode_logic:openKS}
Let \((K,\KpkeOp)\) be an open Kripke structure, and \(\varphi\) a formula of hypernode logic over the same set of variables.  Let $n$ be the number of trace variables in $\varphi$.
Then, \(\Async{K,\KpkeOp}  \models \varphi\) iff \(\Lang(\filter{\varphi}{+}{\stutFreeA_{K,\KpkeOp}^n})\neq \emptyset\).
\end{theorem}

The proof follows from Theorem 1 and 2, and Lemma 3.
This gives us our main result.

\begin{theorem}
\label{thm:hypernode_logic:decidable}
Model checking of hypernode logic over open Kripke structures is decidable.
\end{theorem}

Using our algorithm, 
the running time of model checking a formula of hypernode logic over an open Kripke structure
depends doubly exponentially on the number of variables,  
singly exponentially on the number of worlds of the Kripke structure, 
and singly exponentially on the length of the formula.

\begin{corollary}
\label{thm:upper_bound_mc_hypernonde}
The time complexity of model checking a formula $\varphi$ of hypernode logic with 
$n$ trace variables and $m$ variables, 
over an open Kripke structure with $k$ worlds,
is \(\mathcal{O}(2^{n\cdot k^m}).\)
\end{corollary}

\begin{proof}
The encoding of the open Kripke Structure by a stutter-free automaton has 
\(\mathcal{O}(k^m)\) states.
The determinized stutter-free automaton has 
\(\mathcal{O}(2^{k^m})\) states.
Completing the deterministic stutter-free automaton $2^m$ states.
The $n$-self-composition of the resulting automaton has 
\(\mathcal{O}(2^{n\cdot k^m})\) states.
\end{proof}

  \subsection{Model Checking Hypernode Automata}
  \label{sub:mc_hypernode_automata}

We defined the run of a hypernode automaton for a given action sequence $\actpat$, with each run inducing a slicing of a set of action-labeled traces consistent with $p$. 
To model-check a hypernode automaton \(\fhnaut\) against a pointed Kripke structure \((\Kpke, \ActLab, \KState_0)\) with an action labeling,  we build a finite automaton, called \(\Slice{\Kpke,\ActLab,\KState_0}\), which encodes all slicings of action-labeled traces generated by \((\Kpke, \ActLab, \KState_0)\).
We then reduce the model-checking problem to checking whether the language defined by the composition of \(\Slice{\Kpke,\ActLab,\KState_0}\) with the specification automaton \(\fhnaut\), called \(\Intersect(\fhnaut, K,\ActLab, \KState_0)\), is non-empty.
We depict an overview of the
model-checking algorithm in Fig. \ref{fig:mc_hypernode_aut}.

\begin{figure}
	\centering
	\scalebox{0.70}{ 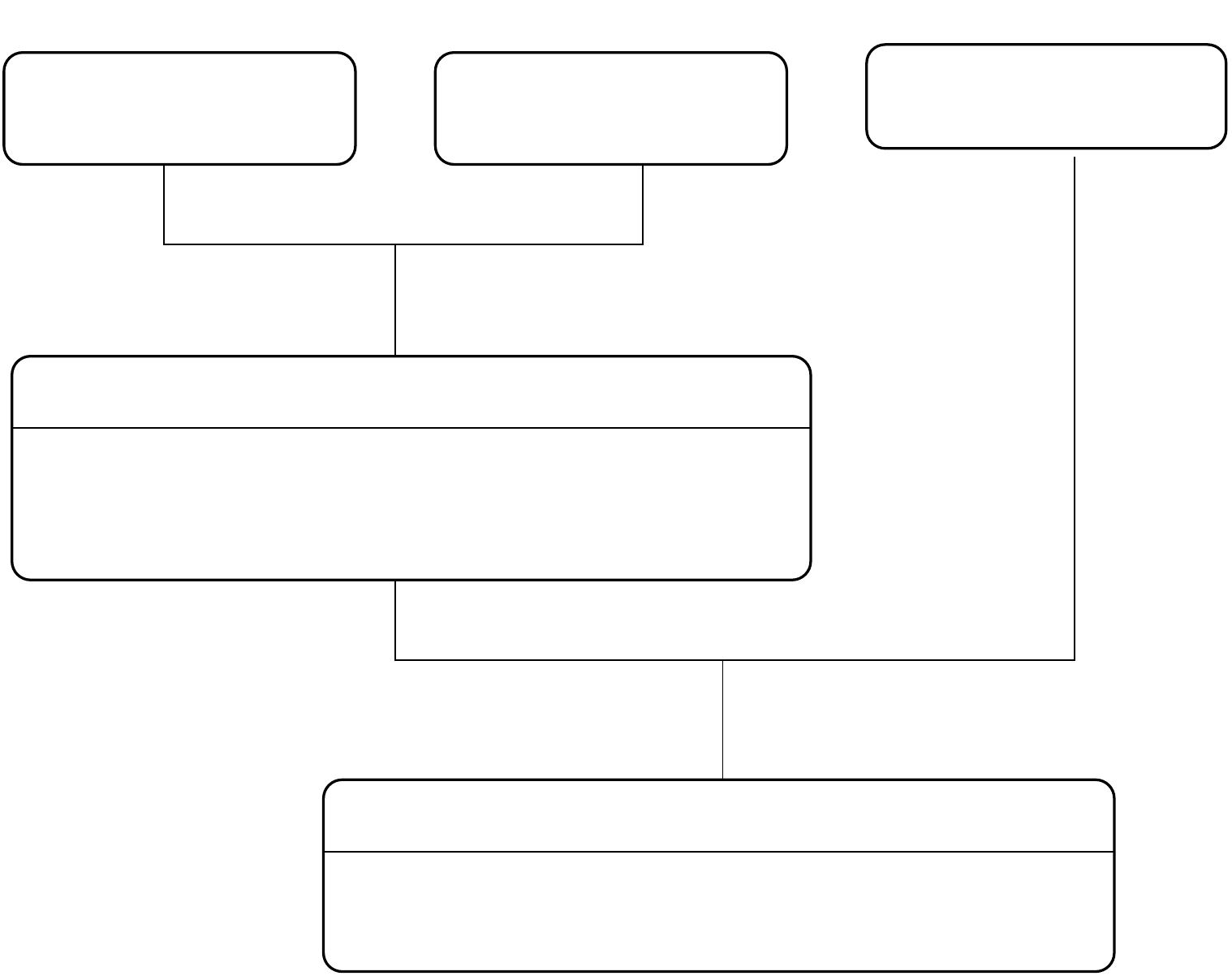 }
	\caption{Model-checking algorithm for hypernode automata with relevant results.}
	\label{fig:mc_hypernode_aut}
\end{figure}

%
%

We start by defining the \emph{slicing} of a given Kripke structure 
\(\KSFull\)
for a given action labeling~\(\ActLab\).
The building blocks of the slicing are Kripke substructures.
A Kripke structure \(\Kpke'\mathop{=} (\KStates', \domain^{\Prop}, \KTransition',  \KLabel')\) is a 
\emph{substructure} of $K$, denoted \(\Kpke' \mathop{\leq} \Kpke\), iff 
\(\KStates' \mathop{\subseteq} \KStates\), and for all worlds \(\KState \mathop{\in} \KStates'\) we have \(\KTransition'(\KState) \mathop{\subseteq} \KTransition(\KState)\) and \(\KLabel'(\KState) \mathop{=} \KLabel(\KState)\).
The
\emph{substructure induced by a transition relation} \(\KTransition' \subseteq \KTransition\)
is  \(\Kpke[\KTransition'] \mathop{=} (W',\Prop, \KTransition', \KLabel(\KStates'))\),
where
$W'=\{\KState, \KState'\,|\, (\KState, \KState')\mathop{\in} \KTransition' \}$.
%
%
%
%
The transition relation defined by \emph{all transitions in a path} of the action-labeled 
Kripke structure \((K,\ActLab)\) from an entry world in \(\KSIn\mathop{\subseteq}\KStates\) to the first step labeled with  \(a \mathop{\in} \Actions\) is:
\vspace{-0.17cm}
\[(\Kpke, \ActLab, \KSIn) \mathop{\downarrow} a \mathop{=} \{(\KState_j, \KState_{j+1}) \,|\,  \KState_0 \varepsilon\, \ldots\,  \KState_{n-1} \varepsilon\,\KState_n a \mathop{\in} \KPaths(K,\ActLab), \KState_0 \mathop{\in} \KSIn \text{ for all } j < n\}.
\]
\vspace{-0.02cm}
\noindent The  
\emph{open substructure induced by \((\Kpke, \ActLab, \KSIn) \mathop{\downarrow} a\)}, 
written \(\mathbb{K}[(\Kpke, \ActLab, \KSIn) \mathop{\downarrow} a]\),
is the open Kripke structure where the Kripke structure is
 \(\Kpke[(\Kpke, \ActLab, \KSIn) \downarrow a]\), 
the set \(\KSIn\) are the entry worlds, and 
the set  
\(\{\KState\,|\, \KState\mathop{\in}\KStates \tAnd \ActLab(\KState, a)\neq \emptyset\}\)
are the exit worlds, 
containing all possible exit points for action \(a\). 

We define the finite automaton \(\Slice{\Kpke,\ActLab,\KState_0}\) and prove, in Lemma \ref{lemma:slicing} below,
that every finite action sequence $\actpat$ defines a unique path in this automaton, 
and the slices of this path contain the same trace segments that are obtained when 
the action sequence $\actpat$ is applied directly to the original pointed, action-labeled Kripke structure.
The states of the automaton \(\Slice{\Kpke,\ActLab,\KState_0}\) are all open substructures 
induced by paths from any choice of entry worlds to an action \(a\in \ActLab\).
Note that there are only finitely many such states.
The transition relation of \(\Slice{\Kpke,\ActLab,\KState_0}\) 
connects, for all actions \(a\), open substructures with exit $a$ and 
open substructures with matching entry worlds.

\begin{definition}
	\label{def:slicing_aut}
	Let \((\Kpke,\KState_0)\) be a pointed Kripke structure with worlds $W$, and let 
	$\ActLab$ be an action labeling for $K$ with actions $A$.
	The slicing \(\Slice{K,\ActLab,\KState_0} \mathop{=} (\States , \SInitial, \Transition)\)
	is a finite automaton where:
	\begin{itemize}
		\item \(\States = \{\mathbb{K}[(\Kpke, \ActLab, \KSIn)  \downarrow a] \ |\ a\in \Actions \tAnd \KSIn \subseteq \KStates\}\) is a set of states with initial states 
		\(\SInitial = \{(\Kpke,(\{\KState_0\},\KSOut)) \in \States\}\);
		\item  \(\Transition\As \States \times \Actions \rightarrow \States\) is a transition function
		s.t.\ \(\Transition((\Kpke, (\KSIn,\KSOut)), a) \mathop{=} (\Kpke', (\KSIn',\KSOut'))\) iff:
		\begin{itemize}
			\item \((\Kpke, (\KSIn,\KSOut))\) exits with action \(a\), that is, for all \(\KState \mathop{\in} \KSOut\) there exists \(\KState'\mathop{\in}\KStates\) such that \(a \mathop{\in} \ActLab(\KState,\KState')\); and 
			\item the set of
			entry worlds \(\KSIn'\) define a maximal subset of the worlds accessible with action \(a\) from the exit worlds in \(\KpkeOp\), that is, for all \((\Kpke'', (\KSIn'',\KSOut'')) \in \States\) that are different from \((\Kpke', (\KSIn',\KSOut'))\):
			if \(\KSIn''\mathop{\subseteq} \{\KState\,|\, a \mathop{\in} \ActLab(\KState',\KState)
			\text{ for some } w'\mathop{\in} \KSOut\}\), then \(\KSIn' \mathop{\not\subseteq} \KSIn''.\)
		\end{itemize} 
	\end{itemize}
\end{definition}

We remark that the transition function \(\Transition\As \States \times \Actions \rightarrow \States\) 
is well-defined, because there is a unique maximal subset for the next entry worlds, given an action \(a\). 
For every two open Kripke substructures, \((\Kpke,(\KSIn, \KSOut)))\) and  \((\Kpke',(\KSIn', \KSOut'))\), their union defines 
\(\mathbb{K}[(\Kpke, \ActLab, \KSIn \cup \KSIn')) \downarrow a]\), 
which is again a state of the slicing.

\begin{lemma}
	\label{lemma:slicing}
	Let \((\Kpke, \KState_0)\) be a pointed Kripke structure,
	and \(\ActLab\) an action labeling for $K$ with actions \(\Actions\).
	For every finite action sequence \(p = a_0 \ldots a_n\) in $\Actions^*$,
	if  \(\Sync{\Kpke, \ActLab, \KState_0}[p] \neq \emptyset\),  
	then $p$ defines a unique run \(\mathbb{K}_0 a_0 \cdots \mathbb{K}_n a_{n} \) of 
	\(\Slice{\Kpke,\ActLab, \KState_0}\) such that for all \(0 \leq i \leq n\),
	\(\KPaths(\mathbb{K}_i) \mathop{=} \KPaths(\Kpke,\ActLab,w_0)(a_0\ldots a_{i-1}, a_i)\).
\end{lemma}
\begin{proof}
	Consider an arbitrary Kripke structure \(\Kpke \mathop{=} (\KStates, \domain^\Prop, \KTransition,  \KLabel)\), world  \(\KState_0 \in \Kpke\) and action labeling \({\ActLab:(\KStates \times \Actions) \rightarrow \KStates}\).
	From the transition function of \(\Slice{\ActLab(\Kpke, \KState_0)}\) being deterministic, it follows that all action sequences \(\actpat \in \Actions^*\) in \((\Kpke, \KState_0)\) with labeling \(\ActLab\), i.e.\ \({\Sync{\ActLab(\Kpke, \KState_0)}[p] \neq \emptyset}\),  define an unique path in \(\Slice{\ActLab(\Kpke, \KState_0)}\).
	
	We still need to prove that paths defined by a slice are the same as slicing the paths generated by \(\ActLab(\Kpke, \KState_0)\), which we prove by induction on the size of the sequence.
	For the base case, for all sequence actions of size~1, \(a_0\), the induced path in \(\Slice{\ActLab(\Kpke, \KState_0)}\) is \(\mathbb{K}_0\), then we need to prove that
	\(\KPaths(\mathbb{K}_0) = \KPaths(\ActLab(\Kpke))(\varnothing, a_0)\).
	By Definition \ref{def:slicing_aut},
	\(\KPaths(\mathbb{K}_0) = \KPaths((\mathbb{K}[(\Kpke, \{\KState_0\}) \toAct a_0]))\).
	And, by definition open substructure induced by \(a\),
	\(\KPaths(\mathbb{K}[(\Kpke, \{\KState_0\}) \toAct a_0]) = \{\KState_0\!\ldots\! \KState_n \, |\,
	(\KState_0, \varepsilon)\! \ldots\!  (\KState_{n-1}, \varepsilon)(\KState_n, a_0)\in\! \KPaths(\Kpke) \}.\)
	Thus, \(\KPaths(\mathbb{K}[(\Kpke, \{\KState_0\}) \toAct a_0]) =\KPaths(\ActLab(\Kpke))(\varnothing, a_0)\).
	
	Now for the induction step, we assume as induction hypothesis (IH) that the statement holds for sequences of size \(n\). 
	Consider now a sequence of size \(n+1\), \(a_0 \ldots a_{n}\). By IH, we know that \(\KPaths(\mathbb{K}_i) = \KPaths(\ActLab(\Kpke))(a_0\ldots a_{i-1}, a_i)\) for all \(0 \leq i < n\).
	We are only missing to prove that \(\KPaths(\mathbb{K}_n) = \KPaths(\ActLab(\Kpke))(a_0\ldots a_{n-1}, a_n)\).
	By IH, we know that \(\KPaths(\mathbb{K}_{n-1})\) and 
	\(\KPaths(\ActLab(\mathbb{K}))(a_0\ldots a_{n-2}, a_{n-1})\) have the same terminal states. Then, it follows that \(\KPaths(\mathbb{K}_{n})\) and 
	\(\KPaths(\ActLab(\mathbb{K}))(a_0\ldots a_{n-1}, a_{n})\) have the same initial states \(\KSIn\). And, from an analogous reasoning from the base case, 
	\(\KPaths(\mathbb{K}_{n}) = \KPaths(\mathbb{K}[(\Kpke, \KSIn)\toAct a_n])= \KPaths(\ActLab(\Kpke))(a_0\ldots a_{n-1}, a_{n})\)
\end{proof}

In a final step, we define a synchronous composition of the slicing automaton defined above 
and the given hypernode automaton $\fhnaut$.
The states of this composition are pairs consisting of open Kripke substructures 
(stemming from the given pointed, action-labeled Kripke structure) 
and formulas of hypernode logic 
(stemming from the hypernode labels of $\fhnaut$).
We mark as final states all pairs where the open Kripke substructure is not a model of the 
hypernode formula.

\begin{definition}
	Let \(\fhnaut\!=\!(\States_h, \initstd, \SLabel, \Transition_h)\) be an hypernode automaton. 
	The \emph{intersection of \(\fhnaut\) with the slicing of a pointed, action-labeled Kripke structure} 
	\((\Kpke, \ActLab, \KState_0)\),
	\(\Slice{K,\ActLab, \KState_0}\! =\! (\States_s , \SInitial_s, \SFinal_s,\Transition_s)\), is the 
	finite automaton
	\(\Intersect(\fhnaut, K, \ActLab, \KState_0)\!\!=\!\!(\States, \SInitial,\SFinal, \Actions, \Transition)\) 
	with
	\begin{itemize}
		\item states \(\States = \{ (\mathbb{K}, \std) \, |\, \mathbb{K} \in \States_s \tAnd {\std \in \States_h} \tAnd \KpkeOp \models \SLabel(\std) \} \  \cup\  \{(\mathbb{K}, \overline{\std}) \, |\, \mathbb{K} \in \States_s \tAnd {\std \in \States_h} \tAnd\linebreak \mathbb{K} \not\models \SLabel(\std)\}\);
		\item initial state \(\SInitial = \{(\mathbb{K}, \initstd) \in \States\ |\ \mathbb{K} \in \SInitial_s \} \cup \{(\mathbb{K}, \overline{\initstd}) \in \States\ |\ \mathbb{K} \in \SInitial_s \}\);
		\item final state \(\SFinal = \{(\mathbb{K}, \overline{\std}) \,|\, (\mathbb{K}, \overline{\std})\in \States\}\);
		\item transition function \(\Transition: \States \times \Actions \rightarrow \States\), where  for all \((\mathbb{K}, \std) \in \States\), we have \(\Transition((\mathbb{K}, \std), a) = \{(\mathbb{K}', \std') \in \States\, |\,  \Transition(\std)= (\std',a)   \tAnd \mathbb{K}' \in \Transition_s(\KpkeOp,a)\} \).
	\end{itemize}
\end{definition}

The finite automaton \(\Intersect(\fhnaut, K, \ActLab, \KState_0)\) reads sequences of actions. 
The notion of run is defined as usual, 
and a run is accepting if it ends in a final state.
The language of the automaton is empty iff it has no accepting run.

\begin{theorem}
\label{thm:mc_hypernode_aut}
Let \((\Kpke,w_0)\) be a pointed Kripke structure with action labeling $\ActLab$.
Let \(\fhnaut\) be a hypernode automaton over the same set of propositions and actions
as $(\Kpke,\ActLab)$.
Then,
\({\Sync{K,\ActLab, \KState_0} \mathop{\in} \Lang(\fhnaut)}\) iff  the language of 
the finite automaton \(\Intersect(\fhnaut, K, \ActLab,  \KState_0)\) is empty.
\end{theorem}

\begin{proof}
	Consider arbitrary \(\Kpke\! =\! (\KStates, \domain^\Prop, \KTransition,  \KLabel)\) over a domain \(\alphabet\), a world \(\KState_0 \in \KStates\) and a action labeling \({\ActLab: (\KStates \times \Actions)\! \rightarrow\! \KStates}\).
	We want to prove that \(\Sync{\ActLab(\Kpke, \KState_0)} \notin \Lang(\fhnaut)\) iff \(\Lang(\Intersect(\fhnaut, \ActLab(\Kpke, \KState_0))) \neq \emptyset.\)
	
	\(\Sync{\ActLab(\Kpke, \KState_0)} \notin \Lang(\fhnaut)\) iff there exists a sequence of actions \(p = a_0 \ldots a_n\) that is in \(\Sync{\ActLab(\Kpke, \KState_0)}\), i.e.\
	\(\Sync{\ActLab(\Kpke, \KState_0)}[p] \neq \emptyset\), 
	and \(\Sync{\ActLab(\Kpke, \KState_0)}[p] \not\models \fhnaut[p]\). Wlog, we can assume that only the last slice does not satisfy the corresponding node in \(\fhnaut\). Let \(\fhnaut[p] = \std_0 a_0 \ldots \std_n a_n\).
	Then, \(\Sync{\ActLab(\Kpke, \KState_0)}[p] \not\models \fhnaut[p]\) iff \((\star)\) for all \(0\leq j <n\), \(\Sync{\ActLab(\Kpke, \KState_0)}(a_0 \ldots a_{j-1}, a_{j}) \models \std_j\) while \(\Sync{\ActLab(\Kpke, \KState_0)}(a_0 \ldots a_{n-1}, a_{n}) \not\models \std_n\).
	By Lemma \ref{lemma:slicing}, \(\Sync{\ActLab(\Kpke, \KState_0)}[p] \neq \emptyset\) defines an unique path in \(\Slice{\ActLab(\Kpke, \KState_0)}\), \(\mathbb{K}_0 a_0 \ldots \mathbb{K}_n a_{n}\) that preserves the path slicing defined by \(\ActLab(\Kpke, \KState_0)\). 
	Thus, from \((\star)\), definition of \(\Intersect(\fhnaut, \ActLab(\Kpke, \KState_0))\) and Lemma \ref{lemma:slicing}, \((\mathbb{K}_0, \std_0) a_0 \ldots (\mathbb{K}_n, \overline{\std_n}) a_n\) defines an accepting run in \(\Intersect(\fhnaut, \ActLab(\Kpke, \KState_0))\). Note that \((\mathbb{K}_n, \overline{\std_n})\) is in \(\Intersect(\fhnaut, \ActLab(\Kpke, \KState_0))\) (and it is final) because \(\Sync{\ActLab(\Kpke, \KState_0)}(a_0 \ldots a_{n-1}, a_{n}) \not\models \std_n\).
\end{proof}

The following theorem puts all results from this section together.

\begin{theorem}
Model checking of hypernode automata over pointed Kripke structures with action labelings 
is decidable.
\end{theorem}

\begin{proof}
We have seen that the model checking of hypernode logic over open Kripke structures
is decidable (Theorem \ref{thm:hypernode_logic:decidable}).
Evaluating \(\Intersect(\fhnaut, K, \ActLab, \KState_0)\) is also decidable.
The main challenge is the slicing of \(\ActLab(\Kpke, \KState_0)\). 
Note that there is a finite number of states that can be in \(\ActLab(\Kpke, \KState_0)\), 
as they are all substructures of the Kripke structure \(\Kpke\). 
\end{proof}

The hardest part of model checking a hypernode automaton over an action-labeled Kripke structure 
is checking the formulas of all hypernodes. 
Therefore, also the running time for model checking hypernode automata is dominated,
as with hypernode logic, 
by a doubly exponential dependency on the number of propositional variables.
Furthermore, our model-checking algorithm depends
singly exponentially on both the size of the Kripke structure and the size of the hypernode automaton. 

\begin{corollary}
\label{thm:upper_bound_mc_hypernonde_aut}
Let \(\Actions\) be a set of actions and \(\Prop\) a set of $m$ propositional variables.
Let \((\Kpke, \KState_0)\) be a pointed Kripke structure over \(\Prop\), and \(\ActLab\) an action labeling 
for $K$ over \(\Actions\). 
Let \(\fhnaut\) be a hypernode automaton over \(\Prop\) and \(\Actions\). 
The time complexity of checking whether \(\Sync{\Kpke,\ActLab,\KState_0)}\mathop{\in} \Lang(\fhnaut)\) 
is \(\mathcal{O}(|\fhnaut|\cdot 2^{|\Actions| + n\cdot|\Kpke|^m})\),  
where \(n\) is the largest number of trace quantifiers that occurs in any hypernode formula in \(\fhnaut\).
\end{corollary}

  \section{Related work}
  \label{sec:related_work}

The first logic studied to express asynchronous hyperproperties
was an extension of \(\mu\)-calculus with explicit quantification over traces, called H\(\mu\) \cite{GutsfeldMO21}.
The trace-quantifier free formulas of  H\(\mu\) are expressively equivalent 
to the parity multi-tape Alternating Asynchronous
Word Automata (AAWA) introduced in \cite{GutsfeldMO21}. 
Both formalisms have highly undecidable model-checking problems.
The undecidability stems from comparing positions in different traces that are arbitrarily far apart, 
over an unbounded number of traces. 
When one of the two dimensions (the distance between compared positions, or the number of traces) 
is given an explicit finite bound, model checking becomes decidable \cite{GutsfeldMO21}. 
In comparison, we achieve decidability by an entirely different means: 
we decouple the progress of different program variables (asynchronicity), 
while allowing resynchronization through automaton-level transitions.

In H\(\mu\) formulas, trace quantifiers always precede time operators, 
while hypernode automata allow a restricted form of quantifier alternation between time 
operators and trace quantifiers.
In particular, the automaton-level transitions correspond to outermost time operators, 
which precede the trace quantifies of hypernode logic formulas, 
whose stutter-reduced prefixing relations correspond to innermost time quantifiers.
As was shown in \cite{BartocciFHNC22}, a change in the quantifier order between trace and 
time quantifiers can cause a synchronous hyperproperty to become asynchronous, and vice versa.
For this reason, 
we conjecture that H\(\mu\) and hypernode automata have incomparable expressive powers.
Consider, for example, the hypernode automaton shown in Figure \ref{fig:exp:hu_ha}, 
which specifies that the asynchronous progress of a propositional variable \(p\) is fully described 
by a finite trace $\pi$ within each slice induced by a repeated action \(a\). 
Each new slice can have a different trace $\pi$ witnessing the asynchronous progress of \(p\).
The length of the traces in each slice is unbounded, 
and as we do not know how many times \(a\) repeats, the number of slices is also unbounded. 
Hence we do not know how many outermost existential trace quantifiers would be needed 
in order to guarantee separate trace witness for each slice. 
Therefore we conjecture that the hyperproperty that is specified by the hypernode automaton 
of Figure \ref{fig:exp:hu_ha} cannot be expressed in H\(\mu\).

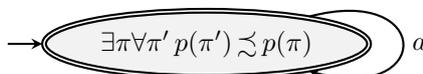
\begin{figure}
\caption{Hypernode automaton specifying that within each slice of a trace set induced by the 
repeated action \(a\), there exists a (possibly different) trace that describes the asynchronous 
progress of the propositional variable \(p\) within the current slice. 
}
\label{fig:exp:hu_ha}
 \centering
	 \begin{tikzpicture}[thick]
	 \node[state, initial, accepting] (s) {\(\exists \traceVar \forall \traceVar' \,  p(\traceVar') \Rel p(\traceVar)\)};
	 \draw
	 	(s) edge[loop right,align=center, looseness = 6] node{\(a\)} (s)
	 	;
	 \end{tikzpicture}
\vspace{-0.5cm}
\end{figure}

Also various extensions of HyperLTL were explored recently in order to support asynchronous 
hyperproperties.
\emph{Stuttering} HyperLTL (HyperLTL$_S$) and \emph{context} HyperLTL (HyperLTL$_C$),
both introduced in~\cite{lics2021},  extend HyperLTL with new operators. 
Stuttering HyperLTL allows the annotation of temporal operators with LTL formulas that describe indistinguishable time sequences. 
In this semantics, the next observable step within a trace is the next point in time at which 
the LTL formula annotating the current time operator has a different truth value. 
Hence time does not progress synchronously across different traces, as it depends on how the valuation of a LTL formula changes within a trace. 
Context HyperLTL introduces a unary modality parameterized by a set of trace variables called \emph{context}. 
The trace variables within a context specify which traces progress as temporal subformulas 
are evaluated (the traces that are not in the context will stutter).
\emph{Asynchronous} HyperLTL (A-HyperLTL)~\cite{cav2021}
extends HyperLTL with quantification over \emph{trajectories}. 
Similar to a context, a trajectory specifies the traces that progress in each evaluation step. 
While the model-checking problems for all of these extensions of HyperLTL are undecidable, 
the authors identify syntactic fragments that support certain asynchronous hyperproperties. 
These decidable fragments adopt restrictions akin to the decidable parts of H\(\mu\).
All of HyperLTL$_S$, HyperLTL$_C$, and A-HyperLTL are subsumed by H\(\mu\) \cite{bozzelli2022expressiveness}.
As we argued that the hypernode automaton 
from Figure \ref{fig:exp:hu_ha} cannot be expressed in H\(\mu\), 
it would neither be expressible in any of the three proposed asynchronous extensions of HyperLTL.

Krebs et al.~\cite{krebs2018team} propose to reinterpret LTL under a so-called \emph{team} semantics.
Team semantics works with sets of variable assignments, 
and the authors introduce both synchronous and asynchronous varieties.
They prove that under asynchronous team semantics, LTL is as expressive as universal HyperLTL 
(where all trace quantifiers are universal quantifiers).
As hypernode logic allows existential quantification over traces, again, our approach is 
orthogonal and expressively incomparable.

In~\cite{LIPIcs.CONCUR.2021.24} the authors introduce HyperATL*, which extends 
alternating-time temporal logic~\cite{AlurHK02}
with strategy quantifiers that bind strategies to trace variables, 
and an explicit construct to resolve games in parallel.
HyperATL* enables the specification of strategic hyperproperties. 
This work is orthogonal to ours as we are interested in \emph{linear-time} asynchronous hyperproperties, 
rather than strategic hyperproperties.

Although many logic-based specification languages have been proposed to express asynchronous hyperproperties, there is a lack of automaton-based approaches to specify such properties. 
Note that AAWA~\cite{GutsfeldMO21} do not support explicit quantification over trace variables. 
The finite-word \emph{hyperautomata} of \cite{BonakdarpourS21} 
constitute a step in this direction by prefixing finite automata with explicit quantification over traces,
but they are limited to \emph{synchronous} hyperproperties.

  \section{Conclusion}
  \label{sec:conclusion}

  We presented a new formalism for specifying hyperproperties of concurrent systems. 
  Our formalism mixes synchronization between different execution traces, 
  expressed as action-labeled transitions of a specification automaton, 
  with asynchronous comparisons between corresponding segments of different traces, 
  expressed as hypernode logic formulas that label the states of the specification automaton. 
  In this way, the specification language of hypernode automata can alternate asynchronous 
  requirements on trace segments of possibly different lengths with synchronization points. 
  Unlike previous formalisms for specifying asynchronous hyperproperties, 
  hypernode automata fully support automatic verification.
  Our model-checking algorithm for hypernode automata is based on an entirely novel technique that 
  introduces stutter-free automata and operations on these automata, 
  thus providing a nice example for the power of automata-theoretic methods in verification.
  
  Besides having a decidable verification problem,
  hypernode automata also represent a genuinely new and useful specification language.
  We demonstrated this by 
  specifying several published variations of observational determinism using hypernode logic,
  by specifying information declassification using hypernode automata, 
  and by arguing that the formalism of hypernode automata is expressively incomparable 
  to various hyperlogics that have been proposed recently for specifying asynchronous 
  hyperproperties.

  The boundary between asynchronicity and synchronicity of trace comparisons can be 
  fine-tuned by introducing variables with compound types, such as boolean arrays, 
  which can be used, for example, 
  to couple the variables of each thread of a multi-threaded program.
  The ramifications of such alphabet variations on hypernode logic and hypernode automata 
  are to be explored in future work.
  There is no shortage of additional topics that follow immediately from the present work 
  but, even if straightforward, require further investigations, 
  including the study of 
  hypernode automata with partial and nondeterministic transition relations, and of 
  hypernode automata with infinitary acceptance conditions (such as hypernode B\"uchi automata), 
  as well as 
  the extension of formal expressiveness studies for hyperproperty specifications
  in order to include hypernode logic and automata, 
  and the presentation of algorithms for solving classical decision problems for 
  hypernode logic and automata other than model checking 
  (such as satisfiability and emptiness). 
  Also the applicability of stutter-free automata in other asynchronous verification contexts 
  (not necessarily concerning hyperproperties) is an interesting question.

\bibliography{main}


\end{document}